\documentclass[11pt]{article}

\usepackage[utf8]{inputenc} 
\usepackage[T1]{fontenc}    
\usepackage{hyperref}       
\usepackage{url}            
\usepackage{booktabs}       
\usepackage{amsfonts}       
\usepackage{nicefrac}       
\usepackage{microtype}      

\usepackage{amstext}
\usepackage{amsmath}
\usepackage{amsthm}
\usepackage{amssymb}
\usepackage{bm}
\usepackage{wrapfig}
\usepackage{amsfonts, amsmath}
\usepackage[utf8]{inputenc}
\usepackage{graphicx}
\usepackage{bbm, comment}
\usepackage{subfigure}
\usepackage{color}
\usepackage{float}
\usepackage{wrapfig}
\usepackage{enumitem}
\usepackage{url}
\usepackage{MnSymbol,wasysym,marvosym,ifsym}

\usepackage{algorithm}
\usepackage[noend]{algpseudocode}

\makeatletter
\def\BState{\State\hskip-\ALG@thistlm}
\makeatother

\usepackage{tabulary}
\usepackage{booktabs}

\usepackage[top=1.in, bottom=1.in, left=1.in, right=1.in]{geometry}
\usepackage[numbers]{natbib}

\newcommand{\E}{\mathbb{E}}

\newcommand{\tru}[0]{\mathbf{T}}

\newcommand{\full}[0]{\text{permutation}}

\newcommand{\info}[0]{\text{informative}}
\newcommand{\iNfo}[0]{\text{Informative}}
\newcommand{\truthful}[0]{\text{truthful}}
\newcommand{\Truthful}[0]{\text{Truthful}}
\newcommand{\focal}[0]{\text{focal}}
\newcommand{\Focal}[0]{\text{Focal}}

\newcommand{\imm}[0]{\text{Mutual Information Paradigm}}
\newcommand{\IMM}[0]{\text{MIP}}

\newtheorem*{theorem*}{Theorem}
\newtheorem{theorem}{Theorem}[section]
\newtheorem{lemma}[theorem]{Lemma}
\newtheorem{fact}[theorem]{Fact}
\newtheorem{claim}[theorem]{Claim}
\newtheorem{assumption}[theorem]{Assumption}

\newtheorem{corollary}[theorem]{Corollary}
\newtheorem{example}[theorem]{Example}
\newtheorem{proposition}[theorem]{Proposition}
\newtheorem{definition}[theorem]{Definition}
\newtheorem{observation}[theorem]{Observation}

\DeclareMathOperator*{\argmax}{arg\,max}

\newcommand{\vecc}{\mathrm{vec}}

\usepackage{xcolor}

\newcommand\gs[1]{{}}
\newcommand\yk[1]{{}}

\begin{document}
\title{An Information Theoretic Framework For Designing Information Elicitation Mechanisms That Reward Truth-telling}  
\author{Yuqing Kong\\ University of Michigan \and Grant Schoenebeck\\ University of Michigan}
\date{}

\begin{titlepage}
\clearpage
  \maketitle
\thispagestyle{empty}

\begin{abstract}
    In the setting where information cannot be verified, we propose a simple yet powerful information theoretical framework---the Mutual Information Paradigm---for information elicitation mechanisms.  Our framework pays every agent a measure of mutual information between her signal and a peer's signal. We require that the mutual information measurement has the key property that any ``data processing'' on the two random variables will decrease the mutual information between them. We identify such information measures that generalize Shannon mutual information.  

Our Mutual Information Paradigm overcomes the two main challenges in information elicitation without verification: (1) how to incentivize effort and avoid agents colluding  to report random or identical responses  (2) how to motivate agents who believe they are in the minority to report truthfully.

Aided by the information measures we found,  (1) we use the paradigm to design a family of novel mechanisms where truth-telling is a dominant strategy and any other strategy will decrease every agent's expected payment  (in the multi-question, detail free, minimal setting where the number of questions is large); (2) we show the versatility of our framework by providing a unified theoretical understanding of existing mechanisms---Peer Prediction [Miller 2005], Bayesian Truth Serum [Prelec 2004], and Dasgupta and Ghosh [2013]---by mapping them into our framework such that theoretical results of those existing mechanisms can be reconstructed easily.

We also give an impossibility result which illustrates, in a certain sense, the the optimality of our framework.

\end{abstract}

\end{titlepage}

\newpage

\section{Introduction}\label{sec:intro}
User feedback requests (e.g Ebay's reputation system and the innumerable survey requests in one's email inbox) are increasingly prominent and important. However, the overwhelming number of requests can lead to low participation rates, which in turn may yield unrepresentative samples. To encourage participation, a system can reward people for answering requests. But this may cause perverse incentives: some people may answer a large of number of questions simply for the reward and without making any attempt to answer accurately. In this case, the reviews the system obtains may be inaccurate and meaningless. Moreover, people may be motivated to lie when they face a potential loss of privacy or can benefit in the future by lying now.

It is thus important to develop systems that motivate honesty. If we can verify the information people provide in the future (e.g prediction markets), we can motivate honesty via this future verification. However, sometimes we need to elicit information without verification since the objective truth is hard to access or even does not exist (e.g. a self-report survey for involvement in crime). In our paper, we focus on the situation where the objective truth is not observable.

Two main challenges in information elicitation without verification area are: 
\begin{enumerate}[noitemsep,topsep=4pt,parsep=2pt,partopsep=4pt]
\item[(1)] How to incentivize effort and avoid colluding agents who report random or identical responses; and
\item[(2)] How to motivate agents who believe they are in the minority to report truthfully.
\end{enumerate}

The typical solution desiderata to the above challenges are: 
\begin{description}[noitemsep,topsep=4pt,parsep=2pt,partopsep=4pt]
\item{\textbf{\emph{(strictly) truthful}}:} truth-telling is a (strict) Bayesian Nash equilibrium.
\item{\textbf{\emph{focal}}:}  the truth-telling equilibrium is paid more than other equilibria in expectation.
\end{description}

 We additionally value  mechanisms that are:
\begin{description}[noitemsep,topsep=4pt,parsep=2pt,partopsep=4pt]
\item{\textbf{\emph{detail free}}:}  require no knowledge of the prior; 
\item{\textbf{\emph{minimal}}:} only require agents to report their information rather than forecasts for other agents' reports; and 
\item{\textbf{\emph{dominantly truthful}}:} truth-telling maximizes the expected payment regardless of the other agents' strategies.  
\end{description}


Two main settings are considered in the information elicitation without verification literature. In the \textbf{\emph{single-question}} setting each agent is asked a single question (e.g. have you ever texted while driving before?) and is assumed to have a \emph{common prior} which means agents who receive the same signal (Y/N) have the same belief for the world. \citet{MRZ05} and \citet{prelec2004bayesian} are two seminal works in this setting.  Another is the \textbf{\emph{multi-question}} setting in which each agent is asked a batch of \emph{apriori similar} questions (e.g.,\ peer grading, or is there a cat in this picture?). \citet{dasgupta2013crowdsourced} is the foundational work in this setting.  Many work~\cite{prelec2004bayesian,witkowski2012robust,witkowski:2014,radanovic2013robust,radanovic2014incentives,zhang2014elicitability,riley2014minimum,faltings2014incentives} have successfully designed truthful and detail free mechanisms.  The design of focal and detail free mechanisms seems to be more complicated.  However, \citet{prelec2004bayesian} is truthful, detail free, and \emph{focal}; and \citet{dasgupta2013crowdsourced} is strictly truthful, detail free, \emph{focal}, and \emph{minimal}.   While these  two prior works have successfully designed focal and detail free mechanisms, their results are typically proved by clever algebraic computations and sometimes lack a deeper intuition, and also fail to extend to important settings.


\subsection{Our Contributions} 

The main contribution of the current paper is to provide a simple yet powerful information theoretic paradigm---the \textbf{\emph{Mutual Information Paradigm}} (Section~\ref{sec:framework})---for designing information elicitation mechanisms that are truthful, focal, and, detail-free.  Moreover, some of the mechanisms based on our paradigm are additionally minimal and dominantly truthful. To the best of our knowledge, ours is the first dominantly truthful information elicitation mechanism without verification.


\paragraph{High Level Techniques and Insights:}
Our framework provides a key insight by distilling the essence of building information elicitation mechanisms to information theory properties.  If we measure ``information'' correctly, any non-truthful strategy will decrease the amount of ``information'' there exists between signals. Therefore, the mechanism should reward every agent according to the amount of information she provides.  That is, the information elicitation mechanism should be ``information-monotone''.

To design ``information-monotone'' mechanisms, we find two families of ``(weakly) information-montone'' information measures---$f$-mutual information and Bregman mutual information---both of which generalize the Shannon mutual information. 

Information theory is not typically used in the information elicitation literature, a key \emph{novelty} of our work is showing how the insights of information theory illuminates the work and challenges in the information elicitation field. 

\paragraph{Applications of Framework:} 

Aided by the Mutual Information Paradigm, 
\begin{enumerate}
    \item In the \textbf{multi-question} setting we exhibit two families of novel mechanisms that are  \emph{dominantly truthful}, detail free, \emph{and} minimal when the number of questions is large---the $f$-mutual information mechanism and the Bregman mutual information mechanism (Section~\ref{sec:md}). Moreover, we show that in the $f$-mutual information mechanism, when any truthful agent changes to play a non-truthful strategy, no matter what strategies other agents play, it decreases every agent's expected payment.  This property implies the mechanism is both dominantly truthful and focal.  We note that dominately truthful is already a stronger equilibrium concept than Bayesian Nash equilibrium which is typically used in the information elicitation mechanisms with no verification literature.
    

    
    We also use a specific $f$-mutual information mechanism to map \citet{dasgupta2013crowdsourced} into our framework with an extra assumption they make and easily reconstruct their main results (Section~\ref{sec:reprovemd}) which apply to the binary choice case and do not require a large number of questions, but also do not have truth-telling as a dominant strategy (only a Bayesian Nash equilibrium).

	\item In the \textbf{single-question} setting, we show that the log scoring rule is an unbiased estimator of Shannon mutual information (Section~\ref{sec:bregman}). Aided by this observation, we map \citet{MRZ05} (Section~\ref{sec:reprovepp}) and \citet{prelec2004bayesian} (Section~\ref{sec:bts}) into our information theoretic framework and allow the easy reconstruction of the theoretical results of \citet{prelec2004bayesian} (Section~\ref{sec:reprovebts}) which is a very important work in the information elicitation literature. We believe our framework highlights important insights of \citet{prelec2004bayesian} which can be extended to other settings.

\end{enumerate}

We also extend our framework to the setting where agents need to exert a certain amount of effort to receive the private information (Section~\ref{sec:effort}) and prove truth-telling is paid \emph{strictly} better than any other non-$\full$ strategy (defined later) in certain settings. Finally, we give impossibility results (Section~\ref{sec:impossibility}) which imply that no truthful detail-free mechanism can pay truth-telling strictly better than the $\full$ strategy in certain settings. This illustrates, in a certain sense, the \emph{optimality} of our framework.

\vspace{5pt}

\subsection{Related Work}
Since \citet{MRZ05} introduced peer prediction, several works follow the peer prediction framework and design information elicitation mechanisms without verification in different settings. In this section, we introduce these work by classifying them into the below four categories:

\emph{(1) Multi-question, Detail Free, Minimal Setting:}
\citet{dasgupta2013crowdsourced} consider a setting where agents are asked to answer multiple a priori similar binary choice questions. They propose a mechanism $M_d$ that pays each agent the correlation between her answer and her peer's answer, and show each agent obtains the highest payment if everyone tells the truth. In retrospect, one can see that our techniques are a recasting and generalization of those of~\citet{dasgupta2013crowdsourced}.
\citet{kamble2015truth} considers both homogeneous and heterogeneous populations and design a mechanism such that truth-telling pays higher than non-informative equilibria in the presence of a large number of a priori similar questions. However, they leave the analysis of other non-truthful equilibria as a open question. \citet{Agarwal:2017:PPH:3033274.3085127} consider a peer prediction mechanism for heterogeneous users.

\yk{reviewer:The remark about [Kamble et. al 2015] is not clear: ‘however their mechanism contains non-truthful equilibria that are paid greater than truth-telling’. They seem to conjecture differently (at least for symmetric equilibria). Previous version: their mechanism may contain non-truthful equilibria that are paid strictly greater than truth-telling}

\emph{(2) Single-question, Detail free, Common Prior, Non-minimal Setting:} ~\citet{prelec2004bayesian} proposes Bayesian Truth Serum (BTS) and the signal-prediction framework for the setting that agents are asked to answer only one question and the mechanism does not know the common prior. \citet{prelec2004bayesian} shows when the number of agents is infinite, the case everyone tells the truth is both an equilibrium and that the total payments agents receive in expectation is at least as high as in any other equilibria. \emph{Logarithmic Peer Truth Serum (PTS)}~\cite{radanovic2015incentive} extends BTS to a slightly different setting involving sensors, but still requires a large number of agents.


The biggest limitation of \citet{prelec2004bayesian} is that the number of agents is assumed to be infinite even to make truth-telling an equilibrium. \citet{witkowski2012robust,witkowski:2014,radanovic2013robust,radanovic2014incentives,zhang2014elicitability,riley2014minimum,faltings2014incentives} successfully weaken this assumption. However, all of these works lack the analysis of non-truthful equilibria. In constrast, the disagreement mechanism designed by \citet{kong2016equilibrium} is truthful and pays truth-telling strictly better than any other symmetric equilibrium if the number of agents is greater than 6.



\emph{(3) Single-question, Known Common Prior, Minimal Setting:}
~\citet{jurca2007collusion,jurcafaltings09} use algorithmic mechanism design to build their own peer prediction style mechanism where truth-telling is paid strictly better than non-truthful \textit{pure} strategies but leave the analysis of mixed strategies as a open question. \citet{frongillo2016geometric} consider the design for robust, truthful and minimal peer prediction mechanisms with the prior knowledge and lack the analysis of non-truthful equilibria. ~\citet{2016arXiv160307319K} modify the peer prediction mechanism such that truth-telling is paid strictly better than any other non-truthful equilibrium. Additionally, they optimize the difference between the truth-telling equilibrium and the next best paying informative equilibrium. However, the mechanism still needs to know the prior and the analysis only works for the case of binary signals.

\yk{reviewer: What does ‘Additionally they optimize the cost their mechanisms needs over a natural space’? Previous version: Additionally, they optimize the cost that their mechanism needs over a natural space.
}



\emph{(4) Other models:}
\citet{Liu:2017:MAP:3033274.3085126} design peer prediction mechanism in the machine learning setting. \citet{mandal2016peer} consider peer prediction mechanism for heterogeneous tasks.




\subsection{Independent Work}\label{sec:independentwork}
Like the current paper\footnote{\citet{2016arXiv160501021K} is the Arxiv version of the current paper without Sections~\ref{sec:effort},~\ref{sec:bregman}.}, \citet{2016arXiv160303151S} also extends \citet{dasgupta2013crowdsourced}'s binary signals mechanism to multiple signals setting.
However, the two works differ both in the specific mechanism and the technical tools employed.

\citet{2016arXiv160303151S} analyze how many questions are needed (whereas we simply assume infinitely many questions).  Like our paper, they also analyze to what extent truth-telling can pay strictly more than other equilibria. Additionally, they show their mechanism does not need a large number of questions when ``the signal correlation structure'' is known (that is the pair-wise correlation between the answers of two questions).  While the current paper does not state such results, we note that the techniques employed are sufficiently powerful to immediately extend to this interesting special case (Appendix~\ref{sec:iw})---when the signal structure is known, it is possible to construct an unbiased estimator for $f$-mutual information of the distribution, when the total variation distance is used to define the $f$-mutual information.
Both papers also show their results generalize~\citet{dasgupta2013crowdsourced}'s.

Moreover, when the number of questions is large, $f$-mutual information mechanism has truth-telling as a dominant strategy while \citet{2016arXiv160303151S} do not.

\label{sec:relatedwork}

\section{Preliminaries}\label{sec:prelim}
\paragraph{General Setting} We introduce the general setting $(n,\Sigma)$ of the mechanism design framework where $n$ is the number of agents and $\Sigma$ is the set of possible private information. Each agent $i$ receives a random private information / signal $\Psi_i:\Omega\mapsto\Sigma$ where $\Omega$ is the underlying sample space. She also has a prior for other agents' private information. 

\yk{Reviewer: add definition of Omega}

Formally, each agent $i$ believes the agents' private information is chosen from a joint distribution $Q_i$ before she receives her private information. Thus, from agent $i$'s perspective, before she receives any private information, the probability that agent $1$ receives $\Psi_1=\sigma_1$, agent 2 receives $\Psi_2=\sigma_2$, ..., agent $n$ receives $\Psi_n=\sigma_n$ is $Q_i(\Psi_1=\sigma_1,\Psi_2=\sigma_2,...,\Psi_n=\sigma_n)$. After she receives her private information based on her prior, agent $i$ will also update her knowledge to a posterior distribution which is the prior conditioned on her private information. Without assuming a common prior, agents may have different priors, that is, $Q_i$ may not equal $Q_j$. We define $\Delta_{\Sigma}$ as the set of all possible probability distributions over $\Sigma$. 


In some situations (e.g. restaurant reviews), agents do not need any effort to receive the private signals when they answer the questions, while in other situations (e.g. peer reviews), agents need to invest a certain amount of effort to receive the private signals. In the current section and Section~\ref{sec:framework}, we assume that agents do not need to invest effort to receive the private signals. Section~\ref{sec:effort} analyzes non-zero effort situations. 





\subsection{Basic Game Theory Concepts}

\begin{definition}[Mechanism]
We define a mechanism $\mathcal{M}$ for a setting $(n,\Sigma)$ as a tuple $\mathcal{M}:=(\mathcal{R},M)$ where $\mathcal{R}$ is a set of all possible reports the mechanism allows, and $M:\mathcal{R}^n\mapsto \mathbb{R}^n$ is a mapping from all agents' reports to each agent's reward.
\end{definition}

The mechanism requires agents to submit a report $r$. For example, $r$ can simply be an agent's private information. In this case, $\mathcal{R}=\Sigma$. We call this kind of mechanism a \emph{minimal} mechanism. We define $\mathbf{r}$ to be a report profile $(r_1,r_2,...,r_n)$ where $r_i$ is agent $i$'s report.

Typically, the strategy of each agent should be a mapping from her received knowledge including her prior and her private signal, to a probability distribution over her report space $\mathcal{R}$. But since all agents' priors are fixed during the time when they play the mechanism, without loss of generality, we omit the prior in the definition of strategy.

\yk{Reviewer: Definition 2.2: Is it ’strategy of M’ or strategy of an agent? Solution: change it already}

\begin{definition}[Strategy]
Given a mechanism $\mathcal{M}$, we define the strategy of each agent in the mechanism $\mathcal{M}$ for setting $(n,\Sigma)$ as a mapping $s$ from $\sigma$ (private signal) to a probability distribution over $\mathcal{R}$.
\end{definition}

We define a strategy profile $\mathbf{s}$ as a profile of all agents' strategies $(s_1,s_2,...,s_n)$ and we say agents play $\mathbf{s}$ if for all $i$, agent $i$ plays strategy $s_i$.

Note that actually the definition of a strategy profile only depends on the setting and the definition of all possible reports $\mathcal{R}$. We will need the definition of a mechanism when we define an equilibrium.

A {\em Bayesian Nash equilibrium} consists of a strategy profile $\mathbf{s} = (s_1, \ldots, s_n)$ such that no agent wishes to change her strategy since other strategy will decrease her expected payment, given the strategies of the other agents and the information contained in her prior and her signal.

\begin{definition}[Agent Welfare]\label{def:agentwelfare}Given a mechanism $\mathcal{M}$, for a strategy profile $\mathbf{s}$, we define the agent welfare of $\mathbf{s}$ as the sum of expected payments to agents when they play $\mathbf{s}$ under $\mathcal{M}$. \end{definition}

\yk{Reviewer: Transition probability: I suppose n does not refer here to the number of agents n? One could add an explanation about the role of theta matrix. Solution: I change n, m to m, m' and add the last sentence.}

\paragraph{Transition Probability} We define a $m\times m'$ transition matrix $M\in \mathbb{R}^{m\times m'}$ as a matrix such that for any $i,j\in [m]\times [m']$, $M_{i,j}\geq 0$ and $\sum_j M_{i,j}=1.$ We define a \emph{permutation transition matrix} $\pi$ as a $m\times m$ permutation matrix.

Given a random variable $X$ with $m$ possible outcomes, by abusing notation a little bit, a $m\times m'$ transition matrix $M$ defines a \textbf{transition probability} $M$ that transforms $X$ to $M(X)$ such that $X':=M(X)$ is a new random variable that has $m'$ possible outcomes where $\Pr[X'=j|X=i]=M_{i,j}$.

If the distribution of $X$ is represented by an $m\times 1$ column vector $\mathbf{p}$, then the distribution over $M(X)$ is $M^T\mathbf{p}$ where $M^T$ is the transpose of $M$.

We can use transition matrices to represent agents' strategies of reporting their private information. Given the general setting $(n,\Sigma)$, for the minimal mechanisms, fixing the priors of the agents, each agent $i$'s strategy $s_i$ can be seen as a transition matrix that transforms her private information $\Psi_i$ to her reported information $\hat{\Psi}_i=s_i(\Psi_i)$. We define \textbf{truth-telling} $\tru$ as the strategy where an agent truthfully reports her private signal. $\tru$ corresponds to an identity transition matrix. 

We say agent $i$ plays a $\full$ strategy if $s_i$ corresponds to a permutation transition matrix. An example is that an agent relabels / permutes the signals and reports the permuted version (e.g. she reports ``good'' when her private signal is ``bad'' and reports ``bad'' when her private signal is ``good''). Note that $\tru$ \footnote{The above definitions of $\tru$ and the permutation strategy are sufficient to analyze the framework. When considering more general settings, we will provide generalized definitions of $\tru$ and the $\full$ strategy.} is a $\full$ strategy as well. We call the strategy profile where all agents play a $\full$ strategy a \emph{$\full$ strategy profile}. Note that in a $\full$ strategy profile, agents may play different $\full$ strategies. When a $\full$ strategy profile is a Bayesian Nash equilibrium, we call such equilibrium a \emph{$\full$ equilibrium}.

\subsection{Mechanism Design Goals}



We hope our mechanisms can be strictly truthful, focal, and even dominantly truthful (see informal definitions in Section~\ref{sec:intro} and formal definitions will be introduced later). Here we propose two additional, stronger equilibrium goals. A mechanism $\mathcal{M}$ is strongly $\focal$ if the truth-telling strategy profile maximizes \emph{every} agent's expected payment among all strategy profiles, while in the focal mechanism, truth-telling maximizes the agent welfare---the sum of agents' expected payment. A mechanism $\mathcal{M}$ is truth-monotone if when any truthful agent changes to play a non-truthful strategy $s$, no matter what strategies other agents play, it decreases every agent's expected payment. Note that the truth-monotone property is stronger than the strongly focal or focal property and it says any non-truthful behavior of any agent will hurt everyone. In addition to the above equilibrium goals, we also hope the mechanism can be minimal and detail free (see definitions in Section~\ref{sec:intro}). 

For the strictness guarantee, it turns out no truthful detail free mechanism can make truth-telling strategy profile be strictly better than any permutation strategy profile. Therefore, the best we can hope is making the truth-telling strategy profile be strictly better than any other non-permutation strategy profile. We give the formal definitions for the equilibrium goals with the strictness guarantee in the below paragraph.

\paragraph{Mechanism Design Goals}\begin{description}[noitemsep,topsep=4pt,parsep=2pt,partopsep=4pt]
\item[(Strictly) $\Truthful$] A mechanism $\mathcal{M}$ is (strictly) truthful if for every agent, $\tru$ (uniquely) maximizes her expected payment given that everyone else plays $\tru$.


\item[(Strictly) Dominantly $\truthful$] A mechanism $\mathcal{M}$ is dominantly truthful if for every agent, $\tru$ maximizes her expected payment no matter what strategies other agents play. A mechanism $\mathcal{M}$ is strictly dominantly truthful if for every agent, if she believes at least one other agent will tell the truth, playing $\tru$ pays her strictly higher than playing  a non-$\full$ strategy. 

\item[(Strictly) $\Focal$] A mechanism $\mathcal{M}$ is (strictly) $\focal$ if the truth-telling equilibrium maximizes the agent welfare among all equilibria (and any other non-$\full$ equilibrium has strictly less agent welfare). 

\item[(Strictly) Strongly $\focal$] A mechanism $\mathcal{M}$ is (strictly) strongly $\focal$ if the truth-telling strategy profile maximizes \emph{every} agent's expected payment among all strategy profiles (and in any other non-$\full$ strategy profile, every agent's expected payment is strictly less). 

\item[(Strictly) Truth-monotone] A mechanism $\mathcal{M}$ is (strictly) truth-monotone if when any truthful agent changes to play a non-truthful strategy $s$, no matter what strategies other agents play, it decreases every agent's expected payment (and strictly decreases every other truthful agent's expected payment if $s$ is a non-$\full$ strategy). 

\end{description}



Section~\ref{sec:impossibility} will show that it is impossible to ask the truth-telling strategy profile to be strictly better than other $\full$ strategy profiles when the mechanism is detail free. Thus, the strictly truth-monotone is the optimal property for equilibrium selection when the mechanism is detail free. 


\subsection{Mechanism Design Tool}\label{sec:tools}

\paragraph{$f$-divergence~\cite{ali1966general,csiszar2004information}} 
$f$-divergence $D_f:\Delta_{\Sigma}\times \Delta_{\Sigma}\rightarrow \mathbb{R}$ is a non-symmetric measure of the difference between distribution $\mathbf{p}\in \Delta_{\Sigma} $ and distribution $\mathbf{q}\in \Delta_{\Sigma} $ 
and is defined to be $$D_f(\mathbf{p},\mathbf{q})=\sum_{\sigma\in \Sigma}
\mathbf{p}(\sigma)f\left( \frac{\mathbf{q}(\sigma)}{\mathbf{p}(\sigma)}\right)$$
where $f(\cdot)$ is a convex function and $f(1)=0$.
Now we introduce the properties of $f$-divergence:



\begin{fact}[Non-negativity~\cite{csiszar2004information}]
For any $\mathbf{p},\mathbf{q} $, $D_f(\mathbf{p},\mathbf{q})\geq 0$ and $D_f(\mathbf{p},\mathbf{q})=0$ if and only if $\mathbf{p}=\mathbf{q}$.
\end{fact}

\begin{fact}[Joint Convexity~\cite{csiszar2004information}]\label{fact:jointconvexity}
For any $0\leq\lambda\leq 1$, for any $\mathbf{p_1},\mathbf{p_2},\mathbf{q}_1,\mathbf{q}_2\in \Delta_{\Sigma}$, $$D_f(\lambda \mathbf{p_1}+(1-\lambda)\mathbf{p_2},\lambda \mathbf{q_1}+(1-\lambda)\mathbf{q_2})\leq \lambda D_f(\mathbf{p_1},\mathbf{q_1})+(1-\lambda)D_f(\mathbf{p_2},\mathbf{q_2}).$$
\end{fact}

\begin{fact}[Information Monotonicity (\cite{ali1966general,liese2006divergences,amari2010information})]\label{lem:im}
For any strictly convex function $f$,  $f$-divergence $D_f(\mathbf{p},\mathbf{q})$ 
satisfies information monotonicity so that for any transition matrix $\theta \in \mathbbm{R}^{|\Sigma| \times |\Sigma|}$, $D_f(\mathbf{p},\mathbf{q})\geq D_f(\theta^T \mathbf{p},\theta^T \mathbf{q})$. 

Moreover, the inequality is strict if and only if there exists $\sigma, \sigma',\sigma''$ such that $\frac{\mathbf{p}(\sigma'')}{\mathbf{p}(\sigma')}\neq \frac{\mathbf{q}(\sigma'')}{\mathbf{q}(\sigma')}$ and $\theta_{\sigma',\sigma} \mathbf{p}(\sigma')>0$, $\theta_{\sigma'',\sigma} \mathbf{p}(\sigma'')>0$.
\end{fact}

The proof is in the appendix for reference.

\begin{definition}\label{def:distinguish}
Given two signals $\sigma',\sigma''\in\Sigma$, we say two probability measures $\mathbf{p},\mathbf{q}$ over $\Sigma$ can \textbf{distinguish} $\sigma',\sigma''\in\Sigma$ if $\mathbf{p}(\sigma')>0$, $\mathbf{p}(\sigma'')>0$ and  $\frac{\mathbf{q}(\sigma')}{\mathbf{p}(\sigma')}\neq \frac{\mathbf{q}(\sigma'')}{\mathbf{p}(\sigma'')}$
\end{definition}

Fact~\ref{lem:im} directly implies

\begin{corollary}\label{cor:strictim}
Given a transition matrix $\theta$ and two probability measures $\mathbf{p},\mathbf{q}$ that can distinguish $\sigma',\sigma''\in\Sigma$, if there exists $\sigma\in \Sigma$ such that $\theta(\sigma',\sigma),\theta(\sigma'',\sigma)>0$, we have $D_f(\mathbf{p},\mathbf{q})> D_f(\theta^T \mathbf{p},\theta^T \mathbf{q})$ when $f$ is a strictly convex function.
\end{corollary}

Now we introduce two $f$-divergences in common use: KL divergence, and Total variation Distance.
\begin{example}[KL divergence]
Choosing $-\log(x)$ as the convex function $f(x)$, $f$-divergence becomes KL divergence $D_{KL}(\mathbf{p},\mathbf{q})=\sum_{\sigma}\mathbf{p}(\sigma)\log\frac{\mathbf{p}(\sigma)}{\mathbf{q}(\sigma)}$
\end{example}

\begin{example}[Total Variation Distance]
Choosing $|x-1|$ as the convex function $f(x)$, $f$-divergence becomes Total Variation Distance $D_{tvd}(\mathbf{p},\mathbf{q})=\sum_{\sigma}|\mathbf{p}(\sigma)-\mathbf{q}(\sigma)|$
\end{example}


\paragraph{Proper scoring rules~\cite{winkler1969scoring,gneiting2007strictly}}
A scoring rule $PS:  \Sigma \times \Delta_{\Sigma} \rightarrow \mathbb{R}$ takes in a signal $\sigma \in \Sigma$  and a distribution over signals $\mathbf{p} \in \Delta_{\Sigma}$ and outputs a real number.  A scoring rule is \emph{proper} if, whenever the first input is drawn from a distribution $\mathbf{p}$, then $\mathbf{p}$ will maximize the expectation of $PS$ over all possible inputs in $\Delta_{\Sigma}$ to the second coordinate. A scoring rule is called \emph{strictly proper} if this maximum is unique. We will assume throughout that the scoring rules we use are strictly proper. Slightly abusing notation, we can extend a scoring rule to be $PS:  \Delta_{\Sigma} \times \Delta_{\Sigma} \rightarrow \mathbb{R}$  by simply taking $PS(\mathbf{p}, \mathbf{q}) = \E_{\sigma \leftarrow \mathbf{p}}(\sigma,  \mathbf{q})$.  We note that this means that any proper scoring rule is linear in the first term. 



\begin{example}[Log Scoring Rule~\cite{winkler1969scoring,gneiting2007strictly}]
Fix an outcome space $\Sigma$ for a signal $\sigma$.  Let $\mathbf{q} \in \Delta_{\Sigma}$ be a reported distribution.
The Logarithmic Scoring Rule maps a signal and reported distribution to a payoff as follows:
$$L(\sigma,\mathbf{q})=\log (\mathbf{q}(\sigma)).$$

Let the signal $\sigma$ be drawn from some random process with distribution $\mathbf{p} \in \Delta_\Sigma$.

Then the expected payoff of the Logarithmic Scoring Rule
$$ \E_{\sigma \leftarrow \mathbf{p}}[L(\sigma,\mathbf{q})]=\sum_{\sigma}\mathbf{p}(\sigma)\log \mathbf{q}(\sigma)=L(\mathbf{p},\mathbf{q})$$

This value will be maximized if and only if $\mathbf{q}=\mathbf{p}$.

\end{example}



\section{An Information Theoretic Mechanism Design Framework}\label{sec:framework}

The original idea of peer prediction~\cite{MRZ05} is based on a clever insight: every agent's information is related to her peers' information and therefore can be checked using her peers' information. Inspired by this, we propose a natural yet powerful information theoretic mechanism design idea---paying every agent the ``mutual information'' between her reported information and her peer's reported information where the ``mutual information'' should be \emph{information-monotone}---any ``data processing'' on the two random variables will decrease the ``mutual information'' between them.

\begin{definition}[Information-monotone mutual information]
We say $MI$ is information-monotone if and only if for any random variables $X:\Omega\mapsto \Sigma_X$ and $Y:\Omega\mapsto \Sigma_Y$:
\begin{description}
\item [Symmetry] $MI(X;Y)=MI(Y;X)$;
\item [Non-negativity] $MI(X;Y)$ is always non-negative and is 0 if $X$ is independent with $Y$;
\item [Data processing inequality] for any transition probability $M\in\mathbb{R}^{|\Sigma_X|\times |\Sigma_X|}$, when $Y$ is independent with $M(X)$ conditioning on $X$, $MI(M(X);Y)\leq MI(X;Y)$.
\end{description}
We say $MI$ is \emph{strictly} information-monotone with respect to a probability measure $P\in \Delta_{\Sigma_X\times \Sigma_Y}$ if when the joint distribution over $X$ and $Y$ is $P$, for any non-permutation $M$, when $Y$ is independent with $M(X)$ conditioning on $X$, $MI(M(X);Y)< MI(X;Y)$. 
\end{definition}

\begin{definition}[Conditional mutual information]\label{def:cmi}
Given three random variables $X,Y,Z$, we define $MI(X;Y|Z)$ as $$\sum_z Pr[Z=z] MI(X;Y|Z=z)$$ where $MI(X;Y|Z=z):=MI(X';Y')$ where $Pr[X'=x,Y'=y]=Pr[X=x,Y=y|Z=z]$.
\end{definition}



\vspace{5pt}

We now provide a paradigm for designing information elicitation mechanisms---the Mutual Information Paradigm.  We warn the reader that \emph{this paradigm represents some ``wishful thinking'' in that is it clear the paradigm cannot compute the payments given the reports}. 

\paragraph{$\imm$ ($\IMM$($MI$))} Given a general setting $(n,\Sigma)$, 
\begin{description}
\item[Report] For each agent $i$, she is asked to provide her private information $\Psi_i$. We denote the actual information she reports as $\hat{\Psi}_i$. 
\item[Payment/Information Score] We uniformly randomly pick a reference agent $j\neq i$ and denote his report as $\hat{\Psi}_j$. Agent $i$ is paid by her information score $$MI(\hat{\Psi}_i;\hat{\Psi}_j)$$ where $MI$ is information-monotone. 
\end{description}


Given a general setting $(n,\Sigma)$, we say $MI$ is \emph{strictly information-monotone with respect to prior $Q$} if for every pair $i,j$, $MI$ is strictly information-monotone with respect to $Q(\Psi_i,\Psi_j)$.

\begin{theorem}\label{thm:imm}
Given a general setting $(n,\Sigma)$, when $MI$ is (strictly) information-monotone (with respect to every agent's prior), the $\imm$ $\IMM$($MI$) is (strictly) dominantly $\truthful$, (strictly) truth-monotone \footnote{We assume that agents do not need effort to receive the private signals in this section and will extend the framework to the non-zero effort setting in Section~\ref{sec:effort}.}.
\end{theorem}

Theorem~\ref{thm:imm} almost immediately follows from the data processing inequality of the mutual information. The key observation in the proof is that \emph{applying any strategy to the information is essentially data processing and thus erodes information}. 

Note that the $\imm$ is not a mechanism since it requires the mechanism to know the full joint distribution over all agents' random private information while agents only report (or even have access to)  a realization / sample of the random private information. Rather, if we design mechanisms such that the payment in the mechanism is an unbiased estimator of the payment in $\imm$, the designed mechanisms will obtain the desirable properties immediately according to Theorem~\ref{thm:imm}. In the future sections, we will see how to design such mechanisms in both the multi-question and single-question settings. 




\begin{proof}
For each agent $i$, for any strategy $s_i$ she plays, comparing with the case she honestly reports $\Psi_i$, her expected information score is $$\sum_{j\neq i}\frac{1}{n-1}MI(\hat{\Psi}_i;\hat{\Psi}_j)=\sum_{j\neq i}\frac{1}{n-1}MI(s_i(\Psi_i);\hat{\Psi}_j)\leq \sum_{j\neq i}\frac{1}{n-1} MI(\Psi_i;\hat{\Psi}_j)$$ since $MI$ is information-monotone. Thus, $\IMM(MI)$ is dominantly truthful when $MI$ is information-monotone.

For the strictness guarantee, we need to show when agent $i$ believes at least one agent tells the truth, for agent $i$, any non-permutation strategy will strictly decrease her expected payment. Let's assume that agent $i$ believes agent $j_0\neq i$ plays $\tru$. When $MI$ is strictly information-monotone with respect to every agent's prior, $MI$ is strictly information-monotone with respect to $Q_i(\Psi_i,\Psi_{j_0})$ as well. Then the inequality of the above formula is strict if agent $i$ plays a non-permutation strategy $s_i$ since $MI(s_i(\Psi_i);\hat{\Psi}_{j_0})=MI(s_i(\Psi_i);{\Psi}_{j_0})<MI(\Psi_i,\Psi_{j_0})$.

Thus, when $MI$ is strictly information-monotone with respect to every agent's prior, $\IMM(MI)$ is strictly dominantly truthful.



Fixing other agents' strategies except agent $k$, for $i\neq k$, agent $i$'s expected payment is

\begin{align*}
    \sum_{j\neq i}\frac{1}{n-1}MI(\hat{\Psi}_i;\hat{\Psi}_j)&=\sum_{j\neq i,k}\frac{1}{n-1}MI(\hat{\Psi}_i;\hat{\Psi}_j)+\frac{1}{n-1}MI(\hat{\Psi}_i;\hat{\Psi}_k)\\&\leq\sum_{j\neq i,k}\frac{1}{n-1}MI(\hat{\Psi}_i;\hat{\Psi}_j)+\frac{1}{n-1}MI(\hat{\Psi}_i;{\Psi}_k).
\end{align*}

Thus, agent $i$'s expected payment decreases when truthful agent $k$ changes to play a non-truthful strategy. For $i=k$, the dominantly truthful property already shows agent $i=k$'s expected payment will decrease when truthful agent $k$ changes to play a non-truthful strategy. Therefore when $MI$ is information-monotone, $\IMM(MI)$ is truth-monotone.

For the strictness guarantee, when $MI$ is strictly information-monotone with respect to every agent's prior, if truthful agent $k$ changes to play a non-$\full$ strategy $s_k$, then a truthful agent $i$'s expected payment will strictly decrease since $ MI({\Psi}_i;s_k({\Psi}_k))<MI({\Psi}_i;{\Psi}_k) $ if $s_k$ is a non-$\full$ strategy and $MI$ is strictly information-monotone.

Therefore, when $MI$ is (strictly) information-monotone (with respect to every agent's prior), $\IMM(MI)$ is (strictly) truth-monotone.
\end{proof}

It remains to design the information-monotone mutual information measure. 

\subsection{$f$-mutual Information} 
Given two random variables $X,Y$, let $\mathbf{U}_{X,Y}$ and $\mathbf{V}_{X,Y}$ be two probability measures where $\mathbf{U}_{X,Y}$ is the joint distribution of $(X,Y)$ and $\mathbf{V}$ is the product of the marginal distributions of $X$ and $Y$. Formally, for every pair of $(x,y)$, $$\mathbf{U}_{X,Y}(X=x,Y=y)=\Pr[X=x,Y=y]\qquad \mathbf{V}_{X,Y}(X=x,Y=y)=\Pr[X=x]\Pr[Y=y].$$ 

If $\mathbf{U}_{X,Y}$ is very different with $\mathbf{V}_{X,Y}$, the mutual information between $X$ and $Y$ should be high since knowing $X$ changes the belief for $Y$ a lot. If $\mathbf{U}_{X,Y}$ equals to $\mathbf{V}_{X,Y}$, the mutual information between $X$ and $Y$ should be zero since $X$ is independent with $Y$. Intuitively, the ``distance'' between $\mathbf{U}_{X,Y}$ and $\mathbf{V}_{X,Y}$ represents the mutual information between them.

\begin{definition}[$f$-mutual information]
The $f$-mutual information between $X$ and $Y$ is defined as $$MI^f(X;Y)=D_f(\mathbf{U}_{X,Y},\mathbf{V}_{X,Y})$$ where $D_f$ is $f$-divergence. 
\end{definition}



\begin{example}[KL divergence and $I(\cdot;\cdot)$]
Choosing $f$-divergence as KL divergence, $f$-mutual information becomes the Shannon (conditional) mutual information~\cite{cover2006elements} $$I(X;Y):=MI^{KL}(X;Y)=\sum_{x,y}\Pr[X=x,Y=y]\log \frac{\Pr[X=x,Y=y]}{\Pr[X=x]\Pr[Y=y]}$$
$$I(X;Y|Z):=MI^{KL}(X;Y|Z)=\sum_{x,y}\Pr[X=x,Y=y,Z=z]\log \frac{\Pr[X=x,Y=y|Z=z]}{\Pr[X=x|Z=z]\Pr[Y=y|Z=z]}.$$
\end{example}

\begin{example}[Total Variation Distance and $MI^{tvd}(\cdot;\cdot)$]
Choosing $f$-divergence as Total Variation Distance, $f$-mutual information becomes $$MI^{tvd}(X;Y):=\sum_{x,y}|\Pr[X=x,Y=y]-\Pr[X=x]\Pr[Y=y]|.$$
\end{example}







For the strictness guarantee, we introduce the following definition:

\begin{definition}[Fine-grained distribution]\label{def:fgd} 
$P\in \Delta_{\Sigma_X\times \Sigma_Y}$ is a fine-grained joint distribution over $X$ and $Y$ if for every two distinct pairs $(x,y),(x',y')$, $U_{X,Y}(X,Y):=P(X,Y)$ and $V_{X,Y}(X,Y):=P(X)P(Y)$ can \textbf{distinguish} (see Definition~\ref{def:distinguish}) $(x,y)$ and $(x',y')$.
\end{definition}

\begin{theorem}[General data processing inequality]\label{lem:imdpi}
When $f$ is strictly convex, $f$-mutual information $MI^f$ is information-monotone and strictly information-monotone with respect to all fine-grained joint distributions over $X$ and $Y$. 
\end{theorem}

\begin{definition}[Fine-grained prior]\label{def:fgp} 
Given general setting $(n,\Sigma)$, $Q$ is fine-grained prior if for every pair $i,j$, $Q(\Psi_i,\Psi_j)$ is a fine-grained joint distribution over $\Psi_i$ and $\Psi_j$. 
\end{definition}

Theorem~\ref{thm:imm} and Theorem~\ref{lem:imdpi} imply the below corollary.

\begin{corollary}\label{coro:fimm}
Given a general setting $(n,\Sigma)$, when $f$ is (strictly) convex (and every agent's prior is fine-grained), the $\imm$ $\IMM$($MI^f$) is (strictly) dominantly $\truthful$, (strictly) truth-monotone. 

\end{corollary}

\begin{proof}[Proof of Theorem~\ref{lem:imdpi}]

We will apply the information monotonicity of $f$-divergence to show the data processing inequality of $f$-mutual information. We first introduce several matrix operations to ease the presentation of the proof.

\begin{definition}[$\vecc$ operator \cite{henderson1981vec}]
The $\vecc$ operator creates a column vector $\vecc(A)$ from a matrix $A$ by stacking the column vectors of $A$.
\end{definition}

\begin{definition}[Kronecker Product \cite{henderson1981vec}]
The Kronecker product of two matrices $A\in \mathbb{R}^{m\times n}$, $B\in \mathbb{R}^{p\times q}$ is defined as the $mp\times nq$ matrix $A\otimes B=\{A_{i,j}B\}=\begin{bmatrix}
    A_{11}B &  \dots  & A_{1n}B \\
     \vdots & \ddots & \vdots \\
    A_{m1}B &  \dots  & A_{mn}B \\
\end{bmatrix}$.
\end{definition}

\begin{fact}[$\vecc$ operator and Kronecker Product \cite{roth1934direct}]\label{fact:vec}
For any matrices $A\in \mathbb{R}^{n_1\times n_2}$, $X\in \mathbb{R}^{n_2\times n_3}$, $B\in \mathbb{R}^{n_3\times n_4}$, $\vecc(AXB)=B^T\otimes A \vecc(X)$.
\end{fact}

Let $X:\Omega\mapsto \Sigma_X$, $Y:\Omega\mapsto \Sigma_Y$ be two random variables. $U_{X,Y}$ and $V_{X,Y}$ can be seen as two $\Sigma_X\times \Sigma_Y$ matrices. Let $M$ be a $|\Sigma_X|\times |\Sigma_X|$ transition matrix. 

We define $\Sigma_{X,Y}$ as $\Sigma_X\times \Sigma_Y$.








Note that the vectorization of the matrix that represents the probability measure over $X$ and $Y$ will not change the probability measure. Thus, $$D_{f}(U_{M(X),Y},V_{M(X),Y})= D_{f}(\vecc(U_{M(X),Y}),\vecc(V_{M(X),Y})).$$ 

We define $I$ as a $|\Sigma_Y|\times |\Sigma_Y|$ identity matrix. For any transition matrix $M$, by simple calculations, we can see the Kronecker product between $M$ and the identity matrix $I$ is a transition matrix as well.

When $Y$ is independent with $M(X)$ conditioning on $X$, for any probability measure $ P\in \Delta_{\Sigma_X\times \Sigma_Y}$ on $X$ and $Y$,
\begin{align}\label{eq:id}
P(M(X)=x',Y=y)&=\sum_x P(M(X)=x'|X=x,Y=y)P(X=x,Y=y)\\ \tag{$Y$ is independent with $M(X)$ conditioning on $X$}
&=\sum_x P(M(X)=x'|X=x)P(X=x,Y=y)
\end{align}

\begin{align}\label{eq:qquad}
MI^f(M(X);Y) =& D_{f}(U_{M(X),Y},V_{M(X),Y})\\ \nonumber
=& D_{f}(\vecc(U_{M(X),Y}),\vecc(V_{M(X),Y}))\\  \tag{equation (\ref{eq:id}), replacing $P$ by $U_{X,Y}$ and $V_{X,Y}$}
=& D_{f}(\vecc(M^T U_{X,Y} I),\vecc(M^T V_{(X),Y} I))\\ \tag{Fact~\ref{fact:vec}}
=& D_{f}(I^T\otimes M^T \vecc(U_{X,Y}),I^T\otimes M^T \vecc(U_{X,Y}))\\  
 \tag{information monotonicity of $f$-divergence}
\leq & D_{f}(\vecc(U_{X,Y}),\vecc(V_{X,Y}))\\ \nonumber
 =&D_{f}(U_{X,Y},V_{X,Y})\\ \nonumber
= & MI^f(X;Y)
\end{align}



Now we show the strictness guarantee. When $M$ is a non-permutation matrix, $\Theta:=(I^T\otimes M^T)^T=M\otimes I$ is a non-permutation matrix as well. Thus there must exist $(x,y),(x',y'),(x'',y'')$ such that both $\Theta((x,y),(x',y'))$ and $\Theta((x,y),(x'',y''))$ are strictly positive where $(x',y')\neq (x'',y'')$. According to the definition of fine-grained prior (see Definition~\ref{def:fgp} ), $U_{X,Y}$ and $V_{X,Y}$ can distinguish $(x',y')$  and $(x',y')$. Then Corollary~\ref{cor:strictim} implies that the inequality in (\ref{eq:qquad}) is strict.

\end{proof}




\section{Extending the Framework to the Zero-one Effort Model}\label{sec:effort}
This section extends the information theoretic framework to the situations (e.g peer grading) where agents need a certain amount of effort to obtain the private signals.

\begin{assumption}[zero-one effort model]\label{assume:effort}
We assume that for each agent $i$, she can either invest full effort $e_i$ to receive the private signal or invest no effort and receive a useless signal which is independent with other agents' private signals. 
\end{assumption}

We define each agent's \emph{utility} as her payment / reward minus her effort. 

For each agent $i$, we say she plays \emph{effort strategy} $\lambda_i$ if and only if she invests full effort with probability $\lambda_i$ and no effort otherwise. We say she plays a pure (mixed) effort strategy if $\lambda_i=0,1$ ($0<\lambda_i<1$). One example of playing a mixed effort strategy is that in the peer grading case, the students only grade some of the homework carefully and give random grades for other homework. 

Given a mechanism $\mathcal{M}$ for setting $(n,\Sigma)$, for each agent $i$, conditioning receiving the private signal $\sigma_i$, we define her \emph{strategy} $s_i$ as a mapping $s$ from $\sigma_i$ (private signal) to a probability distribution over $\mathcal{R}$. Thus, if agent $i$'s effort strategy and strategy are $(\lambda_i,s_i)$, she will invest full effort for $\lambda_i$ fraction of time and for the time she invests full effort, she plays strategy $s_i$. 


In the non-zero effort situation, we focus on incentivizing the full efforts of agents and would like to pursue a stronger property---the (strictly) $\info$-$\truthful$ property---rather than the strictly truthful property since it is not reasonable for every agent to assume everyone else invests full effort and tells the truth in the current situation.




\paragraph{Mechanism Design Goals}
\begin{description}[noitemsep,topsep=4pt,parsep=2pt,partopsep=4pt]
\item[(Strictly) $\iNfo$-$\truthful$] A mechanism $\mathcal{M}$ is (strictly) $\info$-$\truthful$ if for every agent, given that everyone else either always invests no effort or always invests full effort and plays $\tru$, she can maximize her expected utility by either always investing no effort or always investing full effort and playing $\tru$ (and obtain strictly lower utility by always investing full effort and playing a non-$\truthful$ strategy).

\item[(Strictly) Dominantly $\info$] A mechanism $\mathcal{M}$ is dominantly $\info$ if for every agent, no matter what effort strategies and strategies other agents play, she can maximize her expected utility by either always investing no effort or always investing full effort and plays a $\full$ strategy. A mechanism $\mathcal{M}$ is strictly dominantly $\info$ if for every agent, if she believes at least one other agent will always invest full effort and play $\tru$, she will obtain strictly lower expected payment via playing a non-$\full$ strategy than playing $\tru$ when she always invests full effort.

\item[(Strictly) Effort-monotone] A mechanism $\mathcal{M}$ is (strictly) effort-monotone if for every agent, her optimal payment is (strictly) higher if (strictly) more other agents invest full effort and play $\tru$.


\item[Just desserts] A mechanism $\mathcal{M}$ has the just desserts property if each agent's expected payment is always non-negative and her expected payment is zero if she invests no effort. 
\end{description}

In the zero-one effort model, if an agent believes that everyone else invests no effort, it is too much to ask her to invest full effort in an information elicitation mechanism without verification. The below observation shows that if a rational agent believes that  sufficiently many other agents invest full effort and play $\tru$, she will always invest full effort and play $\tru$ as well. 

\begin{observation}
If a mechanism is dominantly $\info$ and strictly effort-monotone, for every rational agent $i$, there is a threshold $\eta_i$ such that if she believes the number of truth-telling agents is above $\eta_i$, she will invest full effort and play a $\full$ strategy. For a fixed  prior, $\eta_i$ is an increasing function of $e_i$.
\end{observation}

The observation is implied by the fact that the utility equals the payment minus the effort. An agent who needs less effort to access the private information has lower threshold.  
\vspace{5pt}

We require an additional property to achieve the above mechanism design goals. 

\yk{Reviewer: P13. Definition 4.3 is too hard to read, could be simplified. The importance of ’Just deserts’ should be explained (scaling? arbitrage?) Solution: do not know how to solve it}

\begin{definition}[Convex mutual information]
We say $MI$ is convex if and only if for any random variables $X_1,X_2, Y$, for any $0\leq \lambda\leq 1$, let $B_{\lambda}$ be an independent Bernoulli variable such that $B_{\lambda}=1$ with probability $\lambda$ and $0$ with probability $1-\lambda$. Let $X$ be a random variable such that if $B_{\lambda}=1$, $X=X_1$, otherwise, $X=X_2$, 
$$MI(X;Y)\leq \lambda MI(X_1;Y)+(1-\lambda) MI(X_2;Y).$$
\end{definition}

\begin{theorem}\label{thm:imm2}
Given a general setting $(n,\Sigma)$ with assumption~\ref{assume:effort}, when $MI$ is convex and (strictly) information-monotone (with respect to every agent's prior), the $\imm$ $\IMM$($MI$)is (strictly) dominantly $\info$, (strictly) effort-monotone and has the just desserts property. 
\end{theorem}
\begin{proof}

In $\IMM$($MI$), agents always have non-negative expected payments since $MI$ is always non-negative. When the other agent invests no effort, according to Assumption~\ref{assume:effort}, her information is useless and independent with other agents' private signals. Thus, her expected payment is 0 since $MI(X;Y)=0$ when $X$ is independent with $Y$.

When agent $i$ plays effort strategy $\lambda_i$ and conditioning on she invests full effort, she plays strategy $s_i$; conditioning on she invests no effort, she reports random variable $X$, her expected utility is 
\begin{align*}
    &\sum_{j\neq i}\frac{1}{n-1} MI(\hat{\Psi}_i;\hat{\Psi}_j )-\lambda_i e_i\\ \tag{$MI$ is convex}
    &\leq \sum_{j\neq i}\frac{1}{n-1} \lambda_i MI(s_i(\Psi_i);\hat{\Psi}_j )+ (1-\lambda_i)MI(X;\hat{\Psi}_j )-\lambda_i e_i\\ \tag{Zero effort implies no information / Assumption~\ref{assume:effort}}
    &=\lambda_i (\sum_{j\neq i}\frac{1}{n-1}MI(s_i(\Psi_i);\hat{\Psi}_j )-e_i)\\ \tag{$MI$ is information-monotone}
    &\leq\lambda_i (\sum_{j\neq i}\frac{1}{n-1}MI(\Psi_i;\hat{\Psi}_j )-e_i)\\
    &\begin{cases}
    \leq \sum_{j\neq i}\frac{1}{n-1}MI(\Psi_i;\hat{\Psi}_j )-e_i & \text{if $\sum_{j\neq i}\frac{1}{n-1}MI(\Psi_i;\hat{\Psi}_j )>e_i$}\\
    \leq 0 & \text{if $\sum_{j\neq i}\frac{1}{n-1}MI(\Psi_i;\hat{\Psi}_j )\leq e_i$}\\
    \end{cases}
\end{align*}

Therefore, no matter what effort strategies and strategies other agents play, agent $i$ can maximize her expected utility by either always investing no effort or always investing full effort and playing a $\full$ strategy.

For the strictness guarantee, we need to show that if some agent $i$ believes at least one other agent tells the truth, then any non-permutation strategy will strictly decrease agent $i$'s expected payment. Let's assume that agent $i$ believes agent $j_0\neq i$ plays $\tru$. When $MI$ is strictly information-monotone with respect to every agent's prior, $MI$ is strictly information-monotone with respect to $Q_i(\Psi_i,\Psi_{j_0})$ as well. Then if agent $i$ always invests full effort and plays a non-permutation strategy $s_i$, $MI(s_i(\Psi_i);\hat{\Psi}_{j_0})=MI(s_i(\Psi_i);{\Psi}_{j_0})<MI(\Psi_i,\Psi_{j_0})$ which implies that in this case, agent $i$'s expected payment is strictly less than her expected payment in the case when she always invest full effort and plays a $\full$ strategy. 

Thus, when $MI$ is convex and (strictly) information-monotone (with respect to every agent's prior), $\IMM$($MI$) is (strictly) dominantly $\info$.

It remains to show the (strictly) effort-monotone property. In order to show the effort-monotone property, we need to show that for every agent $i$, fixing other agents' strategies except agent $k$, if agent $k$ invests full effort and play her optimal strategy $\tru$, agent $i$'s optimal payment increases comparing with the case agent $k$ does not invest full effort. 

When agent $k$ does not always invest full effort, agent $i$'s optimal payment is proportional to $$\sum_{j\neq i} MI(\Psi_i;\hat{\Psi}_j)=\sum_{j\neq i,k} MI(\Psi_i;\hat{\Psi}_j)+MI(\Psi_i;\hat{\Psi}_k)\leq \sum_{j\neq i,k} MI(\Psi_i;\hat{\Psi}_j)+MI(\Psi_i;\Psi_k)$$ when $MI$ is convex and information-monotone. When $MI$ is strictly information-monotone with respect to every agent's prior, with a similar proof as that of the strictly dominantly informative property, we can see the inequality in the above formula is strict.


Thus, when $MI$ is convex and (strictly) information-monotone (with respect to every agent's prior), $\IMM$($MI$) is (strictly) effort-monotone.

\end{proof}

\begin{fact}[Convexity of $f$-mutual information]\label{fact:convexity}
For any $0\leq \lambda\leq 1$, for any random variables $X_1,X_2,Y$, let $B_{\lambda}$ be an independent Bernoulli variable such that $B_{\lambda}=1$ with probability $\lambda$ and $0$ with probability $1-\lambda$. Let $X$ be a random variable such that if $B_{\lambda}=1$, $X=X_1$, otherwise, $X=X_2$, 
$$MI^{f}(X;Y)\leq \lambda MI^{f}(X_1;Y)+(1-\lambda) MI^{f}(X_2;Y).$$
\end{fact}

\begin{proof}
Based on the definition of $X$, 
$$ U_{X,Y}=\lambda U_{X_1,Y} + (1-\lambda) U_{X_2,Y}\qquad V_{X,Y}=\lambda V_{X_1,Y} + (1-\lambda) V_{X_2,Y}.$$

Combining the joint convexity of $D_f$ (Fact~\ref{fact:jointconvexity}) and the fact that $MI^f(X;Y)=D_f(U_{X,Y},V_{X,Y})$, 
$$MI^{f}(X;Y)\leq \lambda MI^{f}(X_1;Y)+(1-\lambda) MI^{f}(X_2;Y).$$

\end{proof}

\begin{corollary}\label{coro:imm2}
Given a general setting $(n,\Sigma)$ with assumption~\ref{assume:effort}, when $f$ is (strictly) convex (and every agent's prior is fine-grained), the $\imm$ $\IMM$($MI^f$) is (strictly) dominantly $\info$, (strictly) effort-monotone and has the just desserts property.

\end{corollary}

\section{Bregman Mutual Information}\label{sec:bregman}
It is naturally to ask whether in addition to $f$-divergence, can we use another commonly used divergence---Bregman divergence $D_{PS}$---to define an information-monotone information measure. Since the general Bregman divergence may not satisfy information monotonicity, the answer is likely to be negative. However, surprisingly, by properly using the Bregman divergence, we can obtain a new family of information measures $BMI^{PS}$ that satisfies almost all information-monotone properties of $f$-mutual information except the symmetry and one half of the data processing inequality. Therefore, by plugging $BMI^{PS}$ into the Mutual Information Paradigm, we may lose the focal property but can preserve the dominantly truthful property.

\paragraph{Bregman Divergence~\cite{bregman1967relaxation,gneiting2007strictly}} Bregman divergence $D_{PS}:\Delta_{\Sigma}\times \Delta_{\Sigma}\rightarrow \mathbb{R}$ is a non-symmetric measure of the difference between distribution $\mathbf{p}\in \Delta_{\Sigma} $ and distribution $\mathbf{q}\in \Delta_{\Sigma} $
and is defined to be $$D_{PS}(\mathbf{p},\mathbf{q})=PS(\mathbf{p},\mathbf{p})-PS(\mathbf{p},\mathbf{q})$$
where $PS$ is a proper scoring rule (see the definition of $PS$ in Section~\ref{sec:tools}). 

\subsection{Bregman Mutual Information}\label{sec:bregmanmutual}
Inspired by the $f$-mutual information, we can first try $D_{PS}(\mathbf{U}_{X,Y},\mathbf{V}_{X,Y})$ to define the Bregman mutual information. However, since the Bregman divergence may not satisfy the information monotonicity, this idea does not work. Intuitively, more information implies a more accurate prediction. Inspired by this intuition, we define Bregman mutual information between $X$ and $Y$ as an accuracy gain---\emph{the accuracy of the posterior $\Pr[\mathbf{Y}|X]$ minus the accuracy of the prior $\Pr[\mathbf{Y}]$}. With this definition, if $X$ changes the belief for $Y$ a lot, then the Bregman mutual information between them is high; if $X$ is independent with $Y$, $\Pr[\mathbf{Y}|X]=\Pr[\mathbf{Y}]$, then the Bregman mutual information between them is zero.

We define $\mathbf{U}_{Y|X=x}$ and $\mathbf{U}_{Y}$ as two probability distribution over $Y$  such that $$\mathbf{U}_{Y|X=x}(Y=y)=\Pr[Y=y|X=x]\qquad \mathbf{U}_{Y}(Y=y)=\Pr[Y=y].$$
\begin{definition}[Bregman mutual information]\label{def:bregmaninformation}
The Bregman mutual information between $X$ and $Y$ is defined as $$BMI^{PS}(X;Y)=\E_{X}D_{PS}(\mathbf{U}_{Y|X},\mathbf{U}_{Y})=\E_{X} PS(\Pr[\mathbf{Y}|X],\Pr[\mathbf{Y}|X])-PS(\Pr[\mathbf{Y}|X],\Pr[\mathbf{Y}]).$$
\end{definition}

\paragraph{Bridging log scoring rule and Shannon mutual information}Inspired by the definition of Bregman mutual information, we will show a novel connection between log scoring rule and Shannon information theory concepts---\emph{the log scoring rule can be used to construct an unbiased estimator of (conditional) Shannon mutual information}. A powerful application of this connection is the information theoretic reconstruction of \citet{prelec2004bayesian} (Section~\ref{sec:reprovebts}). 

The definition of Bregman mutual information says that the accuracy gain measured by a proper scoring rule $PS$ equals the information gain measured by the (conditional) Bregman mutual information $BMI^{PS}$. The below theorem (Theorem~\ref{lem:gain}) shows that we can bridge the log scoring rule and Shannon mutual information by showing the accuracy gain measured by log scoring rule equals the information gain measured by (conditional) Shannon mutual information. Therefore, like $f$-mutual information, Bregman mutual information also generalizes Shannon mutual information (Corollary~\ref{cor:shannon}).


\begin{theorem}[expected accuracy gain = information gain]\label{lem:gain}
For random variables $X,Y,Z$, when predicting $Y$, the logarithm score of prediction $\Pr[\bm{Y}|Z,X]$ minus the logarithm score of prediction $\Pr[\bm{Y}|Z]$
$$\E_{X,Y,Z} L(Y,\Pr[\bm{Y}|Z,X])-L(Y,\Pr[\bm{Y}|Z])=I(X;Y|Z)$$ where $L:\Sigma\times \Delta_{\Sigma}\mapsto \mathbb{R}$ is the log scoring rule and $I(X;Y|Z)$ is the Shannon mutual information between $X$ and $Y$ conditioning on $Z$.
\end{theorem}

\begin{proof}
\begin{align*}
    &\E_{X,Y,Z} L(Y,\Pr[\bm{Y}|Z,X])-L(Y,\Pr[\bm{Y}|Z])\\
    &=\sum_{x,y,z} \Pr[X=x,Y=y,Z=z] \log(\frac{\Pr[Y=y|Z=z,X=x]}{\Pr[Y=y|Z=z]})\\
    &=\sum_{x,y,z} \Pr[X=x,Y=y,Z=z] \log(\frac{\Pr[Y=y,X=x|Z=z]}{\Pr[Y=y|Z=z]\Pr[X=x|Z=z]})\\
    &= I(X;Y|Z)
\end{align*}
\end{proof}

Recall that the conditional mutual information (Definition~\ref{def:cmi}) is defined as $$\sum_z \Pr[Z=z]MI(X;Y|Z=z).$$ Thus, $$BMI^{PS}(X;Y|Z)=\E_{X,Z} PS(\Pr[\mathbf{Y}|X,Z],\Pr[\mathbf{Y}|X,Z])-PS(\Pr[\mathbf{Y}|X,Z],\Pr[\mathbf{Y}|Z])$$ which is the accuracy of posterior $\Pr[\mathbf{Y}|X,Z]$ minus the accuracy of prior $\Pr[\mathbf{Y}|Z]$.  Therefore, Fact~\ref{lem:gain} directly implies Corollary~\ref{cor:shannon}.

\begin{corollary}\label{cor:shannon}
$BMI^{L(\cdot,\cdot)}(X;Y|Z)=I(X;Y|Z)$ where $BMI^{L(\cdot,\cdot)}$ is a Bregman mutual information that chooses Log scoring rule $L(\cdot,\cdot)$ as the proper scoring rule. 
\end{corollary}


\begin{definition}[Quasi Information-monotone mutual information]
We say $MI$ is quasi information-monotone if and only if it is always non-negative and satisfies the data processing inequality for the first entry.
\end{definition}

A quasi information-monotone mutual information may not be symmetric. Thus, even if it satisfies the data processing inequality for the first entry, it may not satisfy the data processing inequality for the second entry which means data processing methods operating on $Y$ may increase $MI(X;Y)$.

\begin{theorem}\label{thm:quasimonotone}
The Bregman mutual information is quasi information-monotone. 
\end{theorem}

Intuitively, more information about $X$ provides a more accurate prediction for random variable $Y$. That is, $\Pr[\mathbf{Y}|M(X)]$ is less accurate than $\Pr[\mathbf{Y}|X]$. We will show the property of the proper scoring rules directly implies the above intuition and then the quasi information-monotonicity of $BMI^{PS}$ follows. 

\begin{proof}
The definition of proper scoring rules implies the non-negativity of Bregman divergence as well as that of Bregman mutual information. 

For any transition probability $M$ that operates on $X$, 

\begin{align*}
    BMI^{PS}(M(X);Y)&=\E_{M(X)}PS(\Pr[\mathbf{Y}|M(X)],\Pr[\mathbf{Y}|M(X)])-PS(\Pr[\mathbf{Y}|M(X)],\Pr[\mathbf{Y}])\\ \tag{$PS$ is linear for the first entry}
    &=\E_{X,M(X)}PS(\Pr[\mathbf{Y}|X,M(X)],\Pr[\mathbf{Y}|M(X)])-PS(\Pr[\mathbf{Y}],\Pr[\mathbf{Y}])\\ \tag{conditioning on $X$, $M(X)$ is independent with $Y$}
    &=\E_{X,M(X)}PS(\Pr[\mathbf{Y}|X],\Pr[\mathbf{Y}|M(X)])-PS(\Pr[\mathbf{Y}],\Pr[\mathbf{Y}])\\ \tag{$PS$ is proper}
    &\leq \E_{X}PS(\Pr[\mathbf{Y}|X],\Pr[\mathbf{Y}|X])-PS(\Pr[\mathbf{Y}],\Pr[\mathbf{Y}])\\
    &=  \E_{X}PS(\Pr[\mathbf{Y}|X],\Pr[\mathbf{Y}|X])-PS(\Pr[\mathbf{Y}|X],\Pr[\mathbf{Y}])\\
    &= BMI^{PS}(X;Y)
\end{align*}
\end{proof}

\begin{theorem}\label{thm:bim1}
Given a general setting $(n,\Sigma)$, when $MI$ is quasi information-monotone, the Mutual Information Paradigm $\IMM$($MI$) is dominantly $\truthful$. 

\end{theorem}

\begin{proof}
For each agent $i$, for any strategy $s_i$ she plays, comparing with the case she honestly reports $\Psi_i$, her expected information score is $$\sum_{j\neq i}MI(\hat{\Psi}_i;\hat{\Psi}_j)=\sum_{j\neq i}MI(s_i(\Psi_i);\hat{\Psi}_j)\leq \sum_{j\neq i} MI(\Psi_i;\hat{\Psi}_j)$$ which is less than if she had reported truthfully since quasi information-monotone $MI$ has data processing inequality for the first entry. 

\end{proof}

\begin{corollary}\label{coro:bim1}
Given a general setting $(n,\Sigma)$, the Mutual Information Paradigm $\IMM$($BMI^{PS}$) is dominantly $\truthful$. 

\end{corollary}

Unfortunately, Bregman mutual information may not be convex. Thus, we cannot extend the analysis for the zero-one effort model to the Bregman mutual information and obtain the dominantly $\info$ property.

\section{Multi-Question, Detail Free, Minimal Setting}\label{sec:multi}
In this section, we introduce the multi-question setting which was previously studied in \citet{dasgupta2013crowdsourced} and \citet{radanovic2014incentives}: $n$ agents are assigned the same $T$ questions (multi-questions). For each question $k$, each agent $i$ receives a \textbf{private signal} $\sigma_i^k\in \Sigma$ about question $k$ and is asked to report this signal. We call this setting $(n,T,\Sigma)$.

We see mechanisms in which agents are not required to report their forecasts for other agents' answer (minimal), and were the mechanism does not know the agents' priors (detail free). Agent $i$ may lie and report $\hat{\sigma}_i^k\neq \sigma_i^k$.  \citet{dasgupta2013crowdsourced} give the following example for this setting: $n$ workers are asked to check the quality of $m$ goods, they may receive signal ``high quality'' or ``low quality''. 

Agents have priors for questions. Each agent $i$ believes agents' private signals for question $k$ are chosen from a joint distribution $Q_i^k$ over $\Sigma^n$. Note that different agents may have different priors for the same question.

We consider the zero-one effort model (Assumption~\ref{assume:effort}) in this setting. For each agent $i$, before she picks her effort strategy, she believes in expectation she will invest effort $e_i^k$ for question $k$. 

In the multi-question setting, people usually make the below assumption:

\begin{assumption}[A Priori Similar and Random Order]\label{assume:ps}
For any $i$, any $k\neq k'$, $Q_i^k=Q_i^{k'}$, $e_i^k=e_i^{k'}$.  Moreover, all questions appear in a random order, independently drawn for each agent.
\end{assumption}

This means agents cannot distinguish each question without the private signal they receive.

We define $(\Psi_1,\Psi_2,...,\Psi_n)$ as the joint random variables such that 
$$\Pr(\Psi_1=\sigma_1,\Psi_2=\sigma_2,...,\Psi_n=\sigma_n)$$ equals the probability that agents $1,2,..,n$ receive private signals $(\sigma_1,\sigma_2,...,\sigma_n)$ correspondingly for a question which is picked uniformly at random. 

We define $(\hat{\Psi}_1,\hat{\Psi}_2,...,\hat{\Psi}_n)$ as the joint random variables such that 
$$\Pr(\hat{\Psi}_1=\hat{\sigma}_1,\hat{\Psi}_2=\hat{\sigma}_2,...,\hat{\Psi}_n=\hat{\sigma}_n)$$ equals the probability that agents $1,2,..,n$ reports signals $(\hat{\sigma}_1,\hat{\sigma}_2,...,\hat{\sigma}_n)$ correspondingly a question which is picked uniformly at random. Note that the joint distribution over $(\hat{\Psi}_1,\hat{\Psi}_2,...,\hat{\Psi}_n)$ depends on the strategies agents play. 

For each question $k$, each agent $i$'s effort strategy is $\lambda_i^k$ and conditioning on that she invests full effort $e_i$, her strategy is $s_i^k$. We say agent $i$ plays a \textbf{consistent strategy} if for any $k,k'$, $\lambda_i^k=\lambda_i^{k'}$ and $ s_i^k=s_i^{k'}$. 


Recall that in the minimal mechanism, the strategy corresponds to a transition matrix. We define \textbf{truth-telling} $\tru$ as the strategy where an agent truthfully reports her private signal for every question. $\tru$ corresponds to the identity matrix. We say agent $i$ plays a \textbf{$\full$} strategy if there exists a permutation transition matrix $\pi$ such that $s_i^k=\pi,\forall k$. Note that a $\full$ strategy is a consistent strategy. We define a \textbf{consistent strategy profile} as the strategy profile where all agents play a consistent strategy. 



With the a priori similar and random order assumption, \citet{dasgupta2013crowdsourced} make the below observation:

\begin{observation}\cite{dasgupta2013crowdsourced}\label{obs:cons}
When questions are a priori similar and agents receive questions in random order (Assumption~\ref{assume:ps}), for every agent, using different strategies for different questions is the same as a mixed consistent strategy.
\end{observation}

With the above observation, it is sufficient to only consider the consistent strategy profiles.

\subsection{The $f$-mutual Information Mechanism and Bregman mutual Information Mechanism}\label{sec:md}
In this section, we give direct applications of the $\imm$ in multi-question setting---the $f$-mutual information mechanism and the Bregman mutual information mechanism. Both of them are a family of mechanisms that can be applied to the non-binary setting / multiple-choices questions which generalize the mechanism in \citet{dasgupta2013crowdsourced} that can only be applied to the binary setting / binary choices questions. Moreover, both the $f$-mutual information mechanism and the Bregman mutual information mechanism are dominantly truthful without considering efforts. Later we will map the mechanism in \citet{dasgupta2013crowdsourced} to a special case of the $f$-mutual information mechanism\footnote{Although $f$-mutual information mechanism requires infinite number of question for clean analysis, with an extra positively correlated assumption for the information structure, we can construct an unbiased estimator for $f$-mutual information of the distribution via only 3 questions (See Section~\ref{sec:independentwork}, Appendix~\ref{sec:iw}).}.

\paragraph{$f$-mutual Information Mechanism $\mathcal{M}_{MI^f}$} Given a multi-question setting $(n,T,\Sigma)$,


\begin{description}
\item[Report] For each agent $i$, for each question $k$, she is asked to provide her private signal $\sigma_i^k$. We denote the actual answer she reports as $\hat{\sigma}_i^k$. 
\item[Payment/Information Score] We arbitrarily pick a reference agent $j\neq i$. We define a probability measure $P$ over $\Sigma\times \Sigma$ such that $T*P(\hat{\Psi}_i=\sigma_i;\hat{\Psi}_j=\sigma_j)$ equals the number of questions that agent $i$ answers $\sigma_i$ and agent $j$ answers $\sigma_j$.

Agent $i$ is paid by her information score $$MI^f(\hat{\Psi}_i;\hat{\Psi}_j)$$ where $(\hat{\Psi}_i;\hat{\Psi}_j)$ draws from the probability measure $P$. 
\end{description}

\begin{theorem}\label{thm:fmutual}
Given a multi-question setting $(n,T,\Sigma)$ with the a priori similar and random order assumption (\ref{assume:ps}), when the number of questions is infinite, $f$ is (strictly) convex (and every agent's prior is fine-grained), (i) without considering efforts, $\mathcal{M}_{MI^f}$ is detail free, minimal, (strictly) dominantly $\truthful$, (strictly) truth-monotone; (ii) with the zero-one effort assumption (\ref{assume:effort}), $\mathcal{M}_{MI^f}$ is detail free, minimal, (strictly) dominantly $\info$, (strictly) effort-monotone and has the just desserts property. 


\end{theorem}

\begin{proof}
We would like to show that the $f$-mutual information mechanism is the same as $\IMM(MI^f)$. Then Corollary~\ref{coro:imm2} directly implies the theorem.

Based on observation~\ref{obs:cons}, it is sufficient to only consider the consistent strategy profiles. When the number of questions is infinite and $\forall i$, agent $i$ play the consistent strategy $\lambda_i,s_i$, $$P(\hat{\Psi}_i=\sigma_i;\hat{\Psi}_j=\sigma_j)=\Pr(\hat{\Psi}_i=\sigma_i;\hat{\Psi}_j=\sigma_j)$$ and $\hat{\Psi}_i=\lambda_i s_i(\Psi_i) + (1-\lambda_i) X_i$ where $X_i$ is a random variable that is independent with all other agents' private signals / strategies.

Therefore, with Assumption \ref{assume:ps}, when the number of questions is infinite, the $f$-mutual information mechanism is the same as $\IMM(MI^f)$ in the multi-question setting. Theorem~\ref{thm:fmutual} follows immediately from Corollary~\ref{coro:fimm},~\ref{coro:imm2}.

\end{proof}

\paragraph{Bregman mutual Information Mechanism $\mathcal{M}_{BMI^{PS}}$} We can define \emph{Bregman mutual information mechanism} via the same definition of $f$-mutual information except replacing $MI^f$ by $BMI^{PS}$.

Corollary~\ref{coro:bim1} directly imply the below theorem. 

\begin{theorem}

Given a multi-question setting $(n,T,\Sigma)$ with the a priori similar and random order assumption (\ref{assume:ps}), when the number of questions is infinite, without considering efforts, the Bregman mutual information mechanism $\mathcal{M}_{BMI^{PS}}$ is detail free, minimal, dominantly $\truthful$. 


\end{theorem}

\subsection{Mapping Dasgupta and Ghosh [2013] into Our Information Theoretic Framework}\label{sec:reprovemd}
This section maps \citet{dasgupta2013crowdsourced} to a special case of $f$-mutual information mechanism (in the binary setting) with the $f$-mutual information using the specific $f$-divergence, total variation distance. With the mapping, we can simplify the proof in \citet{dasgupta2013crowdsourced} to a direct application of our framework.

\subsubsection{Prior Work}

We first state the mechanism $M_d$ and the main theorem in \citet{dasgupta2013crowdsourced}.

\paragraph{Mechanism $M_d$} Agents are asked to report binary signals 0 or 1 for each question. Uniformly randomly pick a reference agent $j$ for agent $i$. We denote $C_i$ as the set of questions agent $i$ answered. We denote $C_j$ as the set of questions agent $j$ answered. We denote $C_{i,j}$ as the set of questions both agent $i$ and agent $j$ answered. For each question $k\in C_{i,j}$ that both agent $i$ and agent $j$ answered, pick subsets $A\subseteq C_i\backslash k,B\subseteq C_j\backslash(k\cup A)$ with $|A|=|B|=d$. If such $A,B$ do not exist, agent $i$'s reward is 0. Otherwise, we define $\bar{\hat{\sigma}}_i^A=\frac{\sum_{l\in A} \hat{\sigma}_i^l}{|A|}$ to be agent $i$'s average answer for subset $A$, $\bar{\hat{\sigma}}_j^B=\frac{\sum_{l\in B} \hat{\sigma}_j^l}{|B|}$ is agent $j$'s average answer for subset $B$.

Agent $i$'s reward for each question $k\in C_{i,j}$ is $$ R_{i,j}^k:=[\hat{\sigma}_i^k*\hat{\sigma}_j^k+(1-\hat{\sigma}_i^k)*(1-\hat{\sigma}_j^k)]-[\bar{\hat{\sigma}}_i^A*\bar{\hat{\sigma}}_j^B+(1-\bar{\hat{\sigma}}_i^A)*(1-\bar{\hat{\sigma}}_j^B)] $$

By simple calculations, essentially agent $i$'s reward for each question $k\in C_{i,j}$ is the correlation between her answer and agent $j$'s answer---$\E[\hat{\Psi}_i \hat{\Psi}_j]-\E[\hat{\Psi}_i]\E[\hat{\Psi}_j]$.

\citet{dasgupta2013crowdsourced} also make an additional assumption:

\begin{assumption}[Positively Correlated]\label{assume:pc}
Each question $k$ has a unknown ground truth $a^k$ and for every agent $i$, with probability greater or equal to $\frac{1}{2}$, agent $i$ receives private signal $a^k$.
\end{assumption}

We succinctly interpret the main results of \citet{dasgupta2013crowdsourced} as well as the results implied by the main results into the below theorem. 

\begin{theorem}\cite{dasgupta2013crowdsourced}\label{md}
Given an multi-question setting $(n,T,\Sigma)$ with the a priori similar and random order assumption (\ref{assume:ps}), the positively correlated assumption (\ref{assume:pc}), and the zero-one effort assumption (\ref{assume:effort}), when $T\geq d+1$, $M_d$ is $\info$-$\truthful$, and strongly focal.  
\end{theorem}

The parameter $d$ can be any positive integer. Larger $d$ will make the mechanism more robust. We will see $M_d$ equals a special case of the $f$-mutual information mechanism only if agent $i,j$'s reported answers are positively correlated. Thus, without considering efforts, $M_d$ is not dominantly truthful while the $f$-mutual information mechanism is. Although $M_d$ only requires a small number of questions, it only applies to binary choice questions, makes an extra assumption, and obtains weaker properties than the $f$-mutual information mechanism. 


\subsubsection{Using Our Information Theoretic Framework to Analyze \citet{dasgupta2013crowdsourced} }\label{sec:dg}

\yk{P20. The claim ’weaker properties than the f-mutual...’ depends on which properties one wants to take into account. [Dasgupta Ghosh 2013] mechanism satisfies strict truthfulness, whereas f-MI does not seem to. Maybe it would be useful to compare the different incentive properties (e.g. say if they are comparable or not...).}

\paragraph{Proof Outline} We will first connect the expected payment in $M_d$ with a specific $f$-mutual information---$ MI^{tvd}$. Then the result follows from the information monotone property of $f$-mutual information. Formally, we use the below claim to show the connection between mechanism $M_d$ and $f$-mutual information mechanism.

\begin{claim}\label{claim:1}[$M_d\approx \mathcal{M}_{MI^{tvd}}$]
With a priori similar and random order assumption, in $M_d$, for every pairs $i,j$, for every reward question $k$, $$\E [R_{i,j}^k]=\frac{1}{2}MI^{tvd}(\Psi_i;\Psi_j)$$ if both of them play $\tru$; $$\E [R_{i,j}^k]\leq \frac{1}{2}MI^{tvd}(\hat{\Psi}_i;\hat{\Psi}_j)$$ if one of them does not play $\tru$.
\end{claim}

Claim~\ref{claim:1} shows the connection between $M_d$ and $\mathcal{M}_{MI^{tvd}}$. The only difference between $M_d$ and $\mathcal{M}_{MI^{tvd}}$ is that for agents $i,j$, when one of the agent does not play $\tru$, the correlation between their reports is upper-bounded by rather than equal to the $tvd$-mutual information. Therefore, in $M_d$, truth-telling is not a dominant strategy. But the information-monotone property of $MI^{tvd}$ still guarantees the informative truthful and strongly focal property of $M_d$. 


\begin{proof}[Proof of Theorem~\ref{md}]

We start to show the informative-truthful property of $M_d$.

For every agent $i$, given that everyone else either invests no effort or plays $\tru$, either her reference agent $j$ invests no effort or plays $\tru$. If agent $j$ invests no effort, agent $i$'s expected payment for each reward question is 
\begin{align*}
    \E [R_{i,j}^k]\leq \frac{1}{2}MI^{tvd}(\hat{\Psi}_i;\hat{\Psi}_j)=0
\end{align*}
no matter what strategy agent $i$ plays. Therefore, we only need to consider the case where agent $j$ plays $\tru$. In this case, we will use a proof that is similar with the proof of Theorem~\ref{thm:imm2} which applies the convexity and information-monotone property of $f$-mutual information. When agent $i$ plays effort strategy $\lambda_i$ and conditioning on her investing full effort, she plays strategy $s_i$, agent $i$'s expected payment for each reward question is 
\begin{align*}
    \E [R_{i,j}^k] \leq \frac{1}{2}MI^{tvd}(\hat{\Psi}_i;\Psi_j)\leq \lambda_i(\frac{1}{2}MI^{tvd}(\Psi_i;\Psi_j))
\end{align*}
since $MI^{tvd}$ is convex and information-monotone. Therefore, given agent $i$ answers $T_i$ questions, let's denote set $\mathcal{J}_i$ as the set of agents such that for any $j\in \mathcal{J}_i$, we can construct size $d$ subsets $A,B$ in $M_d$ for agent $i$ and agent $j$.

The total expected utility of agent $i$ is less than $$\left(\sum_{j\in \mathcal{J}_i,\lambda_j=1, s_j=\tru}\Pr[j]\sum_{k\in C_{i,j}}\lambda_i(\frac{1}{2}MI^{tvd}(\Psi_i;\Psi_j))\right)-\lambda_i T_i e_i.$$ When agent $i$ invests full effort and plays $\tru$, agent $i$'s expected payment for each reward question is $\frac{1}{2}MI^{tvd}(\Psi_i;\Psi_j)$ according to Claim~\ref{claim:1}. The total expected utility of agent $i$ is equal to $$\left(\sum_{j\in \mathcal{J}_i,\lambda_j=1, s_j=\tru}\Pr[j]\sum_{k\in C_{i,j}}(\frac{1}{2}MI^{tvd}(\Psi_i;\Psi_j))\right)- T_i e_i.$$

Note that the upper-bound of the total expected utility of agent $i$ is a linear function of $\lambda_i$. Thus, the agent will maximize her total expected utility by either always investing no effort ($\lambda_i=0$) or always investing full effort ($\lambda_i=1$) and playing $\tru$. The choice depends on how many other agents play $\tru$ and the structure of the reward questions. 

Thus, $M_d$ is informative-truthful. Moreover,

\begin{align*}
    \E [R_{i,j}^k]\leq \frac{1}{2}MI^{tvd}(\hat{\Psi}_i;\hat{\Psi}_j)\leq \frac{1}{2}MI^{tvd}(\Psi_i;\Psi_j)
\end{align*}

Thus, the truth-telling strategy profile maximizes \emph{every} agent's expected payment among all strategy profiles which implies $M_d$ is strongly focal. 

\end{proof}

\paragraph{Proof for Claim~\ref{claim:1}} 



We first show that $$\E [R_{i,j}^k]=\frac{1}{2}MI^{tvd}(\Psi_i;\Psi_j)$$ if both of agents $i,j$ play $\tru$.

Note that by simple calculations, Assumption~\ref{assume:pc} implies that for any $\sigma\in \{0,1\}$, $$\Pr[\Psi_j=\sigma|\Psi_i=\sigma]\geq\Pr[\Psi_j=\sigma],$$ $$\Pr[\Psi_j=\sigma|\Psi_i=\sigma']\leq\Pr[\Psi_j=\sigma],\forall \sigma'\neq \sigma.$$

When both of agents $i,j$ play $\tru$,

\begin{align*}
      \frac{1}{2}MI^{tvd}({\Psi}_i;{\Psi}_j) \tag{Definition of $MI^{tvd}$}
     &=\frac{1}{2}\sum_{\sigma,\sigma'} |\Pr[{\Psi}_i=\sigma,{\Psi}_j=\sigma']-\Pr[{\Psi}_i=\sigma]\Pr[{\Psi}_j=\sigma']|\\ 
     &=\frac{1}{2}\sum_{\sigma,\sigma'} \mathbbm{1}(\sigma=\sigma')\left(\Pr[{\Psi}_i=\sigma,{\Psi}_j=\sigma']-\Pr[{\Psi}_i=\sigma]\Pr[{\Psi}_j=\sigma']\right)\\\tag{Assumption~\ref{assume:pc}}
     &+\mathbbm{1}(\sigma\neq\sigma')\left(\Pr[{\Psi}_i=\sigma]\Pr[{\Psi}_j=\sigma']-\Pr[{\Psi}_i=\sigma,{\Psi}_j=\sigma']\right)\\ \tag{Combining like terms, $\Pr[E]-\Pr[\neg E]=2\Pr[E]-1$}
      &=\sum_{\sigma} \left(\Pr[{\Psi}_i=\sigma,{\Psi}_j=\sigma]-\Pr[{\Psi}_i=\sigma]\Pr[{\Psi}_j=\sigma]\right)\\ \label{eq:rij} \tag{Definition of $R_{i,j}^k$ in $M_d$}
      &=\E[R_{i,j}^k]  
 \end{align*}

The proof of $$\E [R_{i,j}^k]\leq \frac{1}{2}MI^{tvd}(\hat{\Psi}_i;\hat{\Psi}_j)$$ is similar to above proof. We only need to replace ${\Psi}_i$ by $\hat{\Psi}_i$ and change the second equation to greater than, that is, 
 
 \begin{align*}
     &\frac{1}{2}\sum_{\sigma,\sigma'} |\Pr[\hat{\Psi}_i=\sigma,\hat{\Psi}_j=\sigma']-\Pr[\hat{\Psi}_i=\sigma]\Pr[\hat{\Psi}_j=\sigma']|\\  
     &\geq \frac{1}{2}\sum_{\sigma,\sigma'} \mathbbm{1}(\sigma=\sigma')\left(\Pr[\hat{\Psi}_i=\sigma,\hat{\Psi}_j=\sigma']-\Pr[\hat{\Psi}_i=\sigma]\Pr[\hat{\Psi}_j=\sigma']\right)\\\tag{$\sum|x|\geq \sum x $}
     &+\mathbbm{1}(\sigma\neq\sigma')\left(\Pr[\hat{\Psi}_i=\sigma]\Pr[\hat{\Psi}_j=\sigma']-\Pr[\hat{\Psi}_i=\sigma,\hat{\Psi}_j=\sigma']\right).
 \end{align*}

 We have finished the proof of Claim~\ref{claim:1}.

\section{Single-Question, Detail Free, Common Prior, Non-minimal Setting}\label{sec:single}
In the single-question setting, agents are asked to answer a single question which means that the mechanism can only obtain a single sample of all agents' private information $(\Psi_1,\Psi_2,...,\Psi_n)$. We assume that in the single-question setting, agents do not need to invest efforts to receive the private signals. Thus, achieving the $\truthful$ and $\focal$ goals is sufficiently good in the single-question setting. 

\begin{assumption}[Common Prior]
We assume agents have a common prior---it is a common knowledge that agents' private signals for the question are chosen from a common joint distribution $Q$ over $\Sigma^n$.
\end{assumption}

\begin{assumption}[Symmetric Prior]
We assume agents have a symmetric prior $Q$---for any permutation $\pi:[n]\mapsto[n]$, $$ Q(\Psi_1=\sigma_{\pi(1)},\Psi_2=\sigma_{\pi(2)},...,\Psi_n=\sigma_{\pi(n)})=Q(\Psi_1=\sigma_1,\Psi_2=\sigma_2,...,\Psi_n=\sigma_n)  $$
\end{assumption}

Because we will assume that the prior is symmetric, we denote the prior expectation for the fraction of agents who receives private signal $\sigma$ by $q(\sigma)$ and the posterior expectation for the fraction of agents who receives private signal $\sigma$, conditioning on one agent receiving signal $\sigma'$ by $q(\sigma|\sigma')$. We also define $\mathbf{q}=q(\cdot)$ and $\mathbf{q_{\sigma}}=q(\cdot|\sigma)$.


With the above assumptions, agents who receive the same private signals will have the same prediction. 

\begin{assumption}[Informative Prior]
We assume if agents have different private signals, they will have  different expectations for the fraction of at least one signal.  That is for any $\sigma \neq \sigma' $, there exists $\sigma''$ such that $q(\sigma''|\sigma)\neq q(\sigma''|\sigma')$.
\end{assumption}

With the above assumptions, agents who receive different private signals will have different predictions.


\subsection{Warm up: the Original Peer Prediction}\label{sec:priorworkpp}\label{sec:reprovepp}
The original peer prediction method~\cite{MRZ05} makes a prediction on each agent's behalf and pays each agent according to the accuracy of the prediction based on her report. In this section, we will shift the payment of the original peer prediction a little bit and show that in the shifted version, agents are essentially paid by the Bregman mutual information between her information and her peers' information when everyone tells the truth, which matches our Mutual Information Paradigm. 

\subsubsection{Prior Work}

\citet{MRZ05} propose the original peer prediction method to solve the ``motivating minority'' problem. They require the knowledge of the common and symmetric prior such that the mechanism can derive $\mathbf{q}_{\sigma}$ given the private signal $\sigma$.

\paragraph{Original Peer Prediction~\cite{MRZ05}} Each agent $i$ is asked to report her private signal $\sigma_i$. We denote $\hat{\sigma}_i$ as her actual reported signal. The mechanism calculates the prediction $\mathbf{q}_{\hat{\sigma}_i}$ on agent $i$'s behalf and uniformly randomly pick a reference agent $j\neq i$ to pay agent $i$ the accuracy of the prediction $\mathbf{q}_{\hat{\sigma}_i}$---$$PS(\hat{\sigma}_j,\mathbf{q}_{\hat{\sigma}_i})$$

\begin{theorem}~\cite{MRZ05}
The original peer prediction~\cite{MRZ05} is minimal and strictly truthful with known common, symmetric, informative prior.
\end{theorem}

\subsubsection{The Shifted Peer Prediction Method (SPPM)}

Recall that we defined $\mathbf{q}$ as the prior prediction for other agents' received signals before receiving any private signals. The mechanism, which knows the prior, can also derive the prior prediction $\mathbf{q}$.

We shift the original peer prediction subtracting $PS(\hat{\sigma}_j,\mathbf{q})$ from each agent $i$'s original payment. That is, each agent $i$ is paid the accuracy of the posterior minus that of the prior $$PS(\hat{\sigma}_j,\mathbf{q}_{\hat{\sigma}_i})-PS(\hat{\sigma}_j,\mathbf{q}).$$ In this shifted peer prediction mechanism, agents cannot get something for nothing. 

\begin{theorem}
The shifted peer prediction mechanism is minimal and strictly truthful. When everyone tells the truth, each agent $i$'s expected payment is $$BMI^{PS}(\Psi_i;\Psi_j)$$ where agent $j$ is her reference agent. 
\end{theorem}

\begin{proof}
Since $PS(\hat{\sigma}_j,\mathbf{q})$ is independent with agent $i$'s action, the equilibria of the shifted peer prediction mechanism are the same as that of the original peer prediction mechanism. Thus, the shifted peer prediction mechanism is strictly truthful. 

When everyone tells the truth, for every $i$, $\mathbf{q}_{\hat{\sigma}_i}=\mathbf{q}_{{\sigma}_i}=\Pr[\cdot|\Psi_i=\sigma_i]$ where we use the prior $Q$ as the probability measure over all agents' private signals, thus agent $i$ expected payment is 
\begin{align*}
    &\E_{\Psi_i,\Psi_j}PS(\Psi_j,Pr[\bm{\Psi}_j|\Psi_i])-PS(\Psi_j,Pr[\bm{\Psi}_j])\\ \tag{Definition~\ref{def:bregmaninformation}}
    &=BMI^{PS}(\Psi_i;\Psi_j)
\end{align*}
\end{proof}

For SPPM, when agents use other strategies rather than truth-telling, we cannot obtain the result that each agent is paid the mutual information between her reported signal and other agents' reported signal unless both the prior and posterior are updated correctly with the strategy knowledge. However, the mechanism does not know the strategy profile agents will play. Therefore, SPPM is not focal even if we use the log scoring rule. In the next section, we will see the Bayesian truth serum cleverly solves this ``unknown strategy'' problem by setting up new prior and posterior.

\subsection{Mapping Bayesian Truth Serum into Our Information Theoretic Framework}\label{sec:bts}
Bayesian Truth Serum (BTS)~\cite{prelec2004bayesian} rewards the agents whose answer is ``surprisingly popular''. In this section, we will show that in BTS, essentially each agent is paid the mutual information between her information and the aggregated information conditioning a random peer's information which matches our Mutual Information Paradigm. We show this via the connection we found between the log scoring rule and Shannon mutual information---the accuracy gain equals the information gain. Mapping Bayesian Truth Serum into our information theoretic framework substantially simplifies the proof in \citet{prelec2004bayesian} via directly applying the information-monotone property of Shannon mutual information. 

\subsubsection{Prior Work}\label{sec:priorworkbts}
\citet{prelec2004bayesian} proposes the Bayesian Truth Serum mechanism in the single-question setting. In addition to the common prior and the symmetric prior assumptions, two additional assumptions are required:


\begin{assumption}[Conditional Independence]\label{assume:ci}
We define the state of the world as a random variable $W:\Omega\mapsto \Delta_{\Sigma}$ such that given that $W=\omega$, agents' private signals are independently and identically distributed. That is, for every $i$, agent $i$ receives signal $\sigma$ with probability $\omega(\sigma)$. 
\end{assumption}


\begin{assumption}[Large Group]\label{assume:lg}
The number of agents is infinite. 
\end{assumption}

We define a random variable $\hat{W}:\Omega\mapsto \Delta_{\Sigma}$ such that its outcome is the distribution over agents' reported signals. The distribution over $\hat{W}$ dependes on all agents' strategies. With the large group assumption, when agents tell the truth, $\hat{W}=W$.

BTS uses $\hat{W}$ as the posterior distribution and uses agents' forecasts as the prior distribution, and then rewards agents for giving signal reports that are ``unexpectedly common'' with respect to this distribution. Intuitively, an agent will believe her private signal is underestimated by other agents which means she will believe the actual fraction of her own private signal is higher than the average of agents' forecasts. 

Prelec also proposes the signal-prediction framework for the design of detail free mechanism in the single-question setting. 

\paragraph{Signal-prediction framework \cite{prelec2004bayesian}}Given a setting $(n,\Sigma)$ with a symmetric common prior $Q$, the signal-prediction framework defines a game in which each agent $i$ is asked to report his private signal $\sigma_i \in \Sigma$ and his prediction $ \mathbf{p}_i \in \Delta_{\Sigma}$, a distribution over $\Sigma$, where  $\mathbf{p}_i= \mathbf{q}_{\sigma_i}$. For any $\sigma\in \Sigma$, $\mathbf{p}_i(\sigma)$ is agent $i$'s (reported) expectation for the fraction of other agents who has received $\sigma$ given he has received $\sigma_i$. However, agents may not tell the truth. In this framework, the report space $\mathcal{R}=\Sigma \times \Delta_{\Sigma}$. We define a report profile of agent $i$ as $r_i=(\hat{\sigma}_i,\mathbf{\mathbf{\hat{p}}}_i)\in \mathcal{R} $ where $\hat{\sigma_i}$ is agent $i$'s reported signal and $\mathbf{\hat{p}}_i$ is agent $i$'s reported prediction.

We would like to encourage truth-telling $\tru$, namely that agent $i$ reports $\hat{\sigma_i} = \sigma_i,\mathbf{\hat{p}}_i=\mathbf{q}_{\sigma_i}$. To this end, agent $i$ will receive some payment $\nu_i(\hat{\sigma}_i,\mathbf{\hat{p}}_i, \hat{\sigma}_{-i},\mathbf{\hat{p}}_{-i})$ from the mechanism.

\paragraph{Mechanism Bayesian Truth Serum (BTS($\alpha$))~\cite{prelec2004bayesian}} The Bayesian Truth Serum (BTS) follows the signal-prediction framework. Here, we introduce the payment of BTS. Each agent $i$ has two scores: a \textbf{prediction score} and an \textbf{information score}. BTS pays each agent 

$$\text{\textit{prediction score + $\alpha\cdot$ information score}}$$ where $\alpha>1$ To calculate the scores, for every agent $i$, the mechanism chooses a reference agent $j\neq i$ uniformly at random. Agent $i$'s prediction score is $$score_{Pre}(r_i,r_j):=L(\hat{\sigma}_j,\mathbf{\hat{p}}_i)-\log fr(\hat{\sigma}_j|\hat{\bm{\sigma}}_{-j})=\log \mathbf{\hat{p}}_i(\hat{\sigma}_j)-\log fr(\hat{\sigma}_j|\hat{\bm{\sigma}}_{-j})$$ 
Note that only the log scoring rule part $L(\hat{\sigma}_j,\mathbf{\hat{p}}_i)$ is related to agent $i$'s report. Based on the property of the log scoring rule, for agent $i$, in order to maximize her prediction score, the best $\mathbf{\hat{p}}_i(\sigma)$ should be her posterior expectation of the fraction of the agents who \emph{report} $\sigma$ rather than \emph{receive}.

Agent $i$'s information score is $$score_{Im}(r_i,r_j):=\log \frac{fr(\hat{\sigma}_i|\hat{\bm{\sigma}}_{-i})}{\mathbf{\hat{p}}_j(\hat{\sigma}_i)}=\log fr(\hat{\sigma}_i|\hat{\bm{\sigma}}_{-i})-\log \mathbf{\hat{p}}_j(\hat{\sigma}_i) $$ where $fr(\hat{\sigma}_i|\hat{\bm{\sigma}}_{-i})$ is the fraction of all reported signals $\hat{\bm{\sigma}}_{-i}$ (excluding agent $i$) that agree with agent $i$'s reported signal $\hat{\sigma}_i$, which can be seen as the posterior expectation of the fraction of agents who report $\hat{\sigma}_i$ conditioning on all agents' reports, while $\mathbf{\hat{p}}_j(\hat{\sigma}_i)$ is agent $j$'s posterior expectation of that fraction conditioning on agent $j$'s private signal. Intuitively, the signals that actually occur more than other agents believe they will receive a higher information score. 




\vspace{5pt}

Now we restate the main theorem concerning Bayesian Truth Serum:

\begin{theorem}\cite{prelec2004bayesian}\label{bts}
With the common prior, the symmetric prior, the conditional independence, and the large group assumptions, BTS($ \alpha $) is detail free, (i) truthful and (ii) the expected average information score when everyone tells the truth is higher than that in any other equilibrium. Moreover, (iii) for $\alpha>1$, BTS is focal.  
\end{theorem}


\citet{prelec2004bayesian} uses some clever algebraic calculations to prove the main results. In the next section, we will apply our ``accuracy gain=information gain'' observation to map Bayesian Truth Serum \cite{prelec2004bayesian} into our information theoretical framework and show results (ii) and (iii) via applying the data processing inequality of Shannon mutual information. We put \citet{prelec2004bayesian}'s proof for results (i) in appendix since it is already sufficiently simple and not very related to our framework.

\subsubsection{Using Our Information Theoretic Framework to Analyze BTS}\label{sec:reprovebts}


A key observation of BTS is that when agents report the optimal predictions, the average information score is exactly the ``accuracy gain''---the accuracy of the posterior prediction for a \emph{random agent's report} conditioning on \emph{all agents' reports}, minus the accuracy of a \emph{random reference agent $j$}'s posterior prediction for the \emph{random agent's report} conditioning on agent $j$'s private signal. Based on Lemma~\ref{lem:gain}, this accuracy gain equals the Shannon mutual information between a \emph{ random agent's reported signal} and \emph{all agents' reports} conditioning on the \emph{random reference agent $j$}'s private signal $\Psi_j=\sigma_j$. Therefore, the expected information score can be represented as the form of Shannon mutual information. We have similar analysis for the prediction score. We formally state the above observation in Lemma~\ref{lem:ig}. Aided by this lemma, we will show results (ii) and (iii) via applying the information-monotone property of Shannon mutual information.

\begin{lemma}\label{lem:ig}
In BTS, when agents tell the truth, each agent $i$'s expected information score and prediction score are 
$$I(W;\Psi_i|\Psi_j),\qquad -I(W;\Psi_j|\Psi_i)$$ respectively, $\forall j\neq i$. When the agents play a non-truthful equilibrium, we denote random variable $\hat{\Psi}$ as the reported signal of an agent who is picked uniformly at random, the expected average information score and prediction score are $$I(\hat{W};\hat{\Psi}|\Psi_j), \qquad -I(\hat{W};\hat{\Psi}|\Psi_i)$$ respectively, $\forall i,j$.
\end{lemma}

\begin{proof}

When agents tell the truth, each agent $i$'s expected information score is 
\begin{align*}
    &\E_{\Psi_i,\Psi_j,W}L(\Psi_i,Pr[\bm{\Psi}_i|W])-L(\Psi_i,Pr[\bm{\Psi}_i|\Psi_j])\\ \tag{Conditional independence}
    &=\E_{\Psi_i,\Psi_j,W} L(\Psi_i,Pr[\bm{\Psi}_i|W,\Psi_j])-L(\Psi_i,Pr[\bm{\Psi}_i|\Psi_j])\\ \tag{Theorem~\ref{lem:gain} / Expected accuracy gain equals information gain}
    &=I(W;\Psi_i|\Psi_j)
\end{align*}
when she is paired with reference agent $j$. Since the prior is symmetric, $I(W;\Psi_i|\Psi_j)$ is independent of the identity of $j$ if $j\neq i$.

In any equilibrium $\mathbf{s}$, based on the properties of proper scoring rules, each agent $j$ will always maximize his expected prediction score by truthfully reporting his predictions. Moreover, for agent $j$, his reference agent is picked uniformly at random. Therefore, $$\mathbf{\hat{p}}_j(\hat{\sigma})=Pr[\hat{\Psi}=\hat{\sigma}|\Psi_j=\sigma_j]$$ where $\hat{\Psi}$ is the reported signal of an agent who is picked uniformly at random. 

Then we can replace $W,\Psi_i$ by $\hat{W},\hat{\Psi}$ in the above equations and prove that the expected average information score is $$I(\hat{W};\hat{\Psi}|\Psi_j).$$

The analysis for the expected prediction score is the same as the above analysis except that we need to exchange $i$ and $j$.

\end{proof}


\begin{proof}[Proof of Theorem~\ref{bts} (ii), (iii)]


Based on Lemma~\ref{lem:ig}, when agents play an equilibrium, the expected average information score equals 
\begin{align*}
    I(\hat{W};\hat{\Psi}|\Psi_j)&=\sum_{\sigma_j} \Pr[\Psi_j=\sigma_j] I(\hat{W};\hat{\Psi}|\Psi_j=\sigma_j)\\
    \tag{Data processing inequality}
    &\leq \sum_{\sigma_j} \Pr[\Psi_j=\sigma_j] I(\hat{W},W;\hat{\Psi}|\Psi_j=\sigma_j)
\end{align*}

Note that, when the number of agents is infinite, since every agent's strategy is independent with each other, we can see $W$ determines $\hat{W}$ \footnote{When $W=\omega$, $\hat{W}=\frac{1}{n}\sum_i M_i^T \omega$ where $M_i^T$ is the transpose matrix of the transition matrix corresponded to agent $i$'s strategy for signal reporting, and the distribution $\omega$ is represented by a column vector.}. Therefore, 

\begin{align*}
 &\sum_{\sigma_j} \Pr[\Psi_j=\sigma_j] I(\hat{W},W;\hat{\Psi}|\Psi_j=\sigma_j)\\
    &= \sum_{\sigma_j} \Pr[\Psi_j=\sigma_j] I({W};\hat{\Psi}|\Psi_j=\sigma_j)\\
    \tag{Data processing inequality and the symmetric prior assumption}
    &\leq \sum_{\sigma_j} \Pr[\Psi_j=\sigma_j] I(W;\Psi_i|\Psi_j=\sigma_j),\forall i\neq j\\
    &= I(W;\Psi_i|\Psi_j), \forall i\neq j
\end{align*}
Thus, the expected average information score is maximized when everyone tells the truth.


It is left to show for $ \alpha>1 $, in BTS($ \alpha $), the agent-welfare is maximized by truth-telling over all equilibria. Lemma~\ref{lem:ig} shows that when the prior is symmetric, the sum of the expected prediction scores equals the sum of the expected information scores in any equilibrium. Thus, when $ \alpha>1 $, the agent welfare is proportional to the sum of the expected information scores which is maximized by truth-telling over all equilibria. 

\end{proof}

It is natural to ask if we replace the $-log$ in BTS's information score by other convex functions, what property of BTS we can still preserve. The below theorem shows that even though we may not ganrantee the truthful property of BTS, the average expected information score is still monotone with the amount of information for any convex function we use. 

\begin{theorem}
If we replace the information score in BTS by $f( \frac{\hat{p}_j(\hat{\sigma}_i)}{fr(\hat{\sigma}_i|\hat{\bm{\sigma}}_{-i})})$ where $f$ is a convex function, result (ii)---the expected average information score when everyone tells the truth is higher than that in any other equilibrium---is preserved.
\end{theorem}

\begin{proof}
When agents tell the truth, each agent $i$'s expected information score is 
\begin{align*}
    &\E_{\Psi_i,\Psi_j,W}f(\frac{Pr[{\Psi}_i|\Psi_j]}{Pr[{\Psi}_i|W]})\\ \tag{Conditional independence}
    &=\E_{\Psi_i,\Psi_j,W} f(\frac{Pr[{\Psi}_i|\Psi_j]}{Pr[{\Psi}_i|W,\Psi_j]}) \\ &=\E_{\Psi_i,\Psi_j,W} f(\frac{Pr[{\Psi}_i|\Psi_j]Pr[W|\Psi_j]}{Pr[{\Psi}_i,W|\Psi_j]}) \\
    &=MI^f(W;\Psi_i|\Psi_j)
\end{align*}

In any equilibrium $\mathbf{s}$, based on the properties of proper scoring rules, each agent $j$ will always maximize their prediction by truthfully report their predictions, thus, $$\mathbf{\hat{p}}_j(\hat{\sigma})=Pr[\hat{\Psi}=\hat{\sigma}|\Psi_j=\sigma_j].$$ 

Then we can replace $W,\Psi_i$ by $\hat{W},\hat{\Psi}$ in the above equations and prove that the expected average information score is $$MI^f(\hat{W};\hat{\Psi}|\Psi_j).$$

With the similar proof of Theorem~\ref{bts}, the theorem follow immediately from the data processing inequality of $f$-mutual information. 
\end{proof}

\paragraph{BTS and SPPM} Recall that SPPM has the ``unknown strategy'' problem---when agents play other strategies, the mechanism needs to know the strategies to update the prior and posterior correctly. Bayesian Truth Serum\cite{prelec2004bayesian} cleverly solves this problem by asking agents their predictions for other agents' reported signals and setting the new prior as agents' predictions and the new posterior as the prediction conditioning on the aggregated information. Bayesian Truth Serum pays agents according to the accuracy of the new posterior minus that of the new prior. In BTS, both the new prior and new posterior will be automatically updated according to agents' strategies when agents play an equilibrium. Therefore, without the knowledge of the strategies, BTS is still focal.

\section{Impossibility (Tightness) Results}\label{sec:impossibility}

In this section, we will show an impossibility result that implies the optimality of the information theoretical framework. We will see when the mechanism knows no information about the prior profile, no non-trivial mechanism has truth-telling as the \emph{unique} ``best'' equilibrium. Thus, it is too much to ask for a mechanism where truth-telling is paid strictly higher than any other non-truthful equilibrium. The best we can hope is to construct a mechanism where truth-telling is paid strictly higher than all non-truthful equilibria / strategy profiles excluding all \textit{permutation strategy profiles} (Definition~\ref{def:psp}) when the prior is symmetric; or all non-truthful equilibria / strategy profiles excluding all \textit{generalized permutation strategy profiles} (Definition~\ref{def:gpsp}) when the prior may be asymmetric. Because permutation strategies seem unnatural, risky, and require the same amount of effort as truth-telling these are still strong guarantees.

Actually we will show a much more general result in this section that is sufficiently strong to imply the optimality of the framework. Recall that a mechanism is strictly $\focal$ if truth-telling is strictly better than any other strategy profiles excluding generalized permutations strategy profiles. The results of this section imply that no truthful detail free mechanism can pay truth-telling $\tru$ strictly better than all generalized permutations strategy profiles (Definition~\ref{def:gpsp}) \emph{no matter what definition is the truth-telling strategy $\tru$}. 

We omit the prior in the definition of strategy before since we always fix the prior. However, when proving the impossibility results, the prior is not fixed. Therefore, we use the original definition of strategy in this section.  

\begin{definition}[Strategy]
Given a mechanism $\mathcal{M}$, we define the strategy of $\mathcal{M}$ for setting $(n,\Sigma)$ as a mapping $s$ from $(\sigma,Q)$ (private signal and prior) to a probability distribution over $\mathcal{R}$.
\end{definition}

\paragraph{(Generalized) Permutation Strategy Profiles}

A permutation $\pi:\Sigma\mapsto\Sigma$ can be seen as a relabelling of private information. Given two lists of permutations $\bm{\pi}=(\pi_1,\pi_2,...,\pi_n)$, $\bm{\pi}'=(\pi'_1,\pi'_2,...,\pi'_n)$, we define the product of $\bm{\pi}$ and $\bm{\pi}'$ as $$\bm{\pi}\cdot \bm{\pi}':=(\pi_1\cdot\pi'_1,\pi_2\cdot\pi'_2,...,\pi_n\cdot\pi'_n)$$
where for every $i$, $\pi_i\cdot\pi'_i$ is the group product of $\pi_i$ and $\pi'_i$ such that $\pi_i\cdot\pi'_i$ is a new permutation with $\pi_i\cdot\pi'_i(\sigma)=\pi_i(\pi'_i(\sigma))$ for any $\sigma$.

We also define $\bm{\pi}^{-1}$ as $(\pi_1^{-1},\pi_2^{-1},...,\pi_n^{-1})$.

By abusing notation a little, we define $\bm{\pi}:\mathcal{Q}\mapsto\mathcal{Q}$ as a mapping from a prior $Q$ to a \textit{generalized permuted prior} $\bm{\pi}(Q)$ where for any $\sigma_1,\sigma_2,...,\sigma_n\in \Sigma$,
$${\bm{\pi}( Q)}(\sigma_1,\sigma_2,...,\sigma_n)=Q(\pi_1^{-1}(\sigma_1),\pi_2^{-1}(\sigma_2),...,\pi_2^{-1}(\sigma_n))$$ where $\sigma_i$ is the private signal of agent $i$. Notice that it follows that:
$${\bm{\pi}( Q)}(\pi_1(\sigma_1),\pi_2(\sigma_2),...,\pi_2(\sigma_n))=Q(\sigma_1,\sigma_2,...,\sigma_n).$$  Intuitively, $\bm{\pi}(Q)$ is the same as $Q$ after the signals are relabelled according to $\bm{\pi}$.

\begin{definition}[Permutation List Operator on Strategy]
For every agent $i$, given her strategy is $s_i$ and a permutation list $\bm{\pi}$, we define $\bm{\pi}(s_i)$ as the strategy such that  $\bm{\pi}(s_i)(\sigma,Q)=s_i(\pi_i(\sigma),\bm{\pi}(Q))$ for every private information $\sigma$ and prior $Q$.
\end{definition}

\begin{definition}[Permutation List Operator on Strategy Profile]
Given a permutation list $\bm{\pi}$, for any strategy profile $\mathbf{s}$, we define $\bm{\pi}(\mathbf{s})$ as a strategy profile with $\bm{\pi}(\mathbf{s})=(\bm{\pi}(s_1),\bm{\pi}(s_2),...,\bm{\pi}(s_n))$.
\end{definition}

Note that $\bm{\pi}^{-1}\bm{\pi} Q=Q$ which implies $\bm{\pi}^{-1}\bm{\pi}(\mathbf{s})=\mathbf{s}$.

We say $(\pi,\pi,...,\pi)$ is a \emph{symmetric} permutation list for any permutation $\pi$. For convenience, we write $(\pi,\pi,...,\pi)(Q)$ as $\pi (Q)$, $(\pi,\pi,...,\pi)(s)$ as $\pi(s)$ and $(\pi,\pi,...,\pi)(\bm{s})$ as $\pi(\bm{s})$.

We define a permutation strategy (profile) and then give a generalized version of this definition.

\begin{definition}[Permutation Strategy ]
We define a strategy $s$ as a permutation strategy if there exists a permutation $\pi$ such that $s=\pi(\tru)$.
\end{definition}

\begin{definition}[Permutation Strategy Profile ]\label{def:psp}
We define a strategy profile $\mathbf{s}$ as a permutation strategy profile if there exists a permutation $\pi$ such that $\mathbf{s}=\pi(\tru)$.
\end{definition}

\begin{definition}[Generalized Permutation Strategy Profile]\label{def:gpsp}
We define a strategy profile $\mathbf{s}$ as a generalized permutation strategy profile if there exists a permutation list $\bm{\pi}=(\pi_1,\pi_2,...,\pi_n)$ such that $\mathbf{s}=\bm{\pi}(\tru)=(\pi_1,\pi_2,...,\pi_n)(\tru)$.
\end{definition}

\subsection{Tightness Proof}

\begin{definition}
Given a prior profile $\mathbf{Q}=(Q_1,Q_2,...,Q_n)$ and a strategy profile $\mathbf{s}=(s_1,s_2,...,s_n)$, and a mechanism $\mathcal{M}$,  for every agent $i$, we define $$\nu^{\mathcal{M}}_{i}(n,\Sigma,\mathbf{Q},\mathbf{s})$$ as agent $i$'s ex ante expected payment when agents play $\mathbf{s}$ and all agents' private information is drawn from $Q_{i}$ that is, from agent $i$'s viewpoint.
\end{definition}

The impossibility result is stated as below:

\begin{proposition}\label{prop:impossibility}
Let $\mathcal{M}$ be a mechanism that does not know the prior profile, then for any strategy profile $s$, and any permutation list $\bm{\pi}$:\\
(1) $\mathbf{s}$ is a strict Bayesian Nash equilibrium of $\mathcal{M}$ for any prior profile iff $\bm{\pi}(\mathbf{s})$ is a strict Bayesian Nash equilibrium of $\mathcal{M}$ for any prior profile.\\
(2) For every agent $i$, there exists a prior profile $\mathbf{Q}$ such that $\nu^{\mathcal{M}}_{i}(n,\Sigma,\mathbf{Q},\mathbf{s})\leq \nu^{\mathcal{M}}_{i}(n,\Sigma,\mathbf{Q},\bm{\pi}(\mathbf{s})) $.

Additionally, if the mechanism knows the prior is symmetric, the above results only hold for any symmetric permutation list $(\pi,\pi,...,\pi)$. 
\end{proposition}



Proposition~\ref{prop:impossibility} implies

\begin{corollary}\label{cor:im}
Let $\mathcal{M}$ be a truthful mechanism, given truth-telling strategy $\tru$,
when $\mathcal{M}$ knows no information about the prior profile of agents, if there exists a permutation list $\bm{\pi}$ such that $\bm{\pi}(\tru)\neq \tru$, $\tru$ cannot be always paid strictly higher than all generalized permutation strategy profiles.

Additionally, if the mechanism knows the prior is symmetric, the above results only hold for any symmetric permutation list $(\pi,\pi,...,\pi)$ and all permutation strategy profiles. 
\end{corollary}

We note that the requirement that there exists $\bm{\pi}$ such that ${\bm{\pi}}(\tru)\neq \tru$ only fails for very trivial mechanisms where the truthfully reported strategy does not depend on the signal an agent receives.

The key idea to prove this theorem is what we refer to as \textbf{Indistinguishable Scenarios}:

\begin{definition}[Scenario]
	We define a scenario for the setting $(n,\Sigma)$ as a tuple $(\mathbf{Q},\mathbf{s})$ where $\mathbf{Q}$ is a prior profile, and $\mathbf{s}$ is a strategy profile. 
\end{definition}

Given mechanism $\mathcal{M}$, for any scenario $A=(\mathbf{Q}_A,\mathbf{s}_A)$, we write $\nu^{\mathcal{M}}_{i_A}(n,\Sigma,A)$ as agent $i_A$'s ex ante expected payment when agents play $\mathbf{s}_A$ and all agents' private signals are drawn from $Q_{i_A}$.

For two scenarios $A=(\mathbf{Q}_A,\mathbf{s}_A)$, $B=(\mathbf{Q}_B,\mathbf{s}_B)$ for setting $(n,\Sigma)$, let $\sigma_A:=(\sigma_{1_A},\sigma_{2_A},...,
\sigma_{n_A})$ be agents $(1_A,2_A,...,n_A)$' private signals respectively in scenario $A$, $\sigma_B:=(\sigma_{1_B},\sigma_{2_B},...,
\sigma_{n_B})$ be agents $(1_B,2_B,...,n_B)$' private signals respectively in scenario $B$. 

\begin{definition}[Indistinguishable Scenarios]\label{indistinguishable}
	We say two scenarios $A,B$ are indistinguishable $A\approx B$ if there is a coupling of the random variables $\sigma_A$ and $\sigma_B$ such that $\forall i$, $s_{i_A}(\sigma_{i_A},Q_{i_A})=s_{i_B}(\sigma_{i_B},Q_{i_B})$ and agent $i_A$ has the same belief about the world as agent $i_B$, in other words, for every $j$, $Pr(\hat{r}_{j_A}=\hat{r}|\sigma_{i_A},\mathbf{Q}_{A},\mathbf{s}_A)=Pr(\hat{r}_{j_B}=\hat{r}|\sigma_{i_B},\mathbf{Q}_{B},\mathbf{s}_B)$ $\forall \hat{r}\in\mathcal{R}$. 
\end{definition}

Now we will prove two properties of indistinguishable scenarios which are the main tools in the proof for our impossibility result.

\begin{observation}
	If $(\mathbf{Q}_A,\mathbf{s}_A)\approx (\mathbf{Q}_B,\mathbf{s}_B)$, then (i) for any mechanism $\mathcal{M}$,  $\mathbf{s}_A$ is a (strict) equilibrium for the prior profile $\mathbf{Q}_A$ iff $\mathbf{s}_B$ is a (strict) equilibrium for the prior profile $\mathbf{Q}_B$. (ii) $\forall i$, $\nu^{\mathcal{M}}_{i_A}(n,\Sigma,A)=\nu^{\mathcal{M}}_{i_B}(n,\Sigma,B)$
\end{observation}

At a high level, (1) is true since any reported profile distribution that agent $i_A$ can deviate to, agent $i_B$ can deviate to the same reported profile distribution as well and obtain the same expected payment as agent $i_A$.

Formally, we will prove the $\Rightarrow$ direction in (1) by contradiction. The proof of the other direction will be similar. Consider the coupling for $\sigma_A,\sigma_B$ mentioned in the definition of indistinguishable scenarios. For the sake of contradiction, assume there exists $i$ and $\sigma_{i_B}$ such that $\hat{r}'\neq s_{i_B}(\sigma_{i_B},Q_{i_A})$ is a best response for agent $i_B$. Since agent $i_A$ has the same belief about the world as agent $i_B$ and $s_{i_A}(\sigma_{i_A},Q_{i_A})=s_{i_B}(\sigma_{i_B},Q_{i_B})$, $\hat{r}'\neq s_{i_A}(\sigma_{i_A},Q_{i_A})$ is a best response to agent $i_A$ as well, which is a contradiction to the fact that $\mathbf{s}_A$ is a strictly equilibrium for prior $Q_{i_A}$. 

To gain intuition about (2), consider the coupling again.  For any $i$, agent $i_A$ reports the same thing and has the same belief for the world as agent $i_B$, which implies the expected payoff of agent $i_A$ is the same as agent $i_B$. (2) follows.

Now we are ready to prove our impossibility result:
\begin{proof}[of Proposition~\ref{prop:impossibility}]
	We prove part (1) and part (2) separately. 
	
	\paragraph{Proof of Part (1)} 
	
	Let $A:=(\mathbf{Q},\mathbf{s}),B:=(\bm{\pi}^{-1}(\mathbf{Q}),\bm{\pi}(\mathbf{s}))$. We will show that for any strategy profile $\mathbf{s}$ and any prior $Q$, $A\approx B$. Based on our above observations, part (1) immediately follows from that fact. 
	
	To prove $(Q,\mathbf{s})\approx (\bm{\pi}^{-1}Q,\bm{\pi}(\mathbf{s}))$, for every $i$, we can couple $(\sigma_1,\sigma_2,...,\sigma_n)$ with
	
	$(\pi_{1}^{-1}(\sigma_1),\pi_{2}^{-1}(\sigma_2),..,\pi_{n}^{-1}(\sigma_n))$ where $(\sigma_1,\sigma_2,...,\sigma_n)$ is drawn from $Q_{i}$. It is a legal coupling since $$ {\bm{\pi}^{-1}(Q_{i})}(\pi_{1}^{-1}(\sigma_1),\pi_{2}^{-1}(\sigma_2),..,\pi_{n}^{-1}(\sigma_n))={Q_{i}}(\sigma_1,\sigma_2,...,\sigma_n)$$ according to the definition of $\bm{\pi}^{-1}(Q)$. 
	
	Now we show this coupling satisfies the condition in Definition~\ref{indistinguishable}. First note that\\ $\bm{\pi}(s_i)(\bm{\pi}^{-1}(\sigma_i),\bm{\pi}^{-1}(Q))=s_i(\sigma_i,Q)$. Now we begin to calculate $Pr(\hat{r}_{j_B}=\hat{r}|\sigma_{i_B},\mathbf{Q}_{B},\mathbf{s}_B)$
	\begin{align}
	\label{im0}
	Pr(\hat{r}_{j_B}=\hat{r}|\sigma_{i_B},\mathbf{Q}_{B},\mathbf{s}_B)=&Pr(\hat{r}_{j_B}=\hat{r}|\pi_i^{-1}(\sigma_{i_A}),\bm{\pi}^{-1} (Q_{j_A}), \bm{\pi}(\mathbf{s}_A))\\ \label{im1}
	=&\sum_{\sigma'}{\bm{\pi}^{-1} (Q_{i_A})}(\sigma'|\pi_i^{-1}(\sigma_{i_A}))Pr({\bm{\pi}(s_{j_A})(\sigma',\bm{\pi}^{-1} (Q_{j_A}))}=\hat{r})\\ \label{im2}
	=&\sum_{\sigma'}{\bm{\pi}^{-1} (Q_{i_A})}(\sigma'|\pi_i^{-1}(\sigma_{i_A}))Pr(s_{j_A}(\bm{\pi} (\sigma'),\bm{\pi} \bm{\pi}^{-1} (Q_{j_A}))=\hat{r})\\ \label{im3}
	= &\sum_{\sigma'}{Q_{i_A}}(\pi_j(\sigma')|\sigma_{i_A})Pr( s_{j_A}(\bm{\pi} (\sigma'), Q_{j_A})=\hat{r})\\ \label{im4}
	= & \sum_{\sigma''}{Q_{i_A}}(\sigma''|\sigma_{i_A})Pr(s_{j_A}(\sigma'', Q_{i_A})=\hat{r})\\ \label{im5}
	=& Pr(\hat{r}_{j_A}=\hat{r}|\sigma_{i_A},\mathbf{Q}_{A},s_A)
	\end{align}
	
	From (\ref{im0}) to (\ref{im1}): To calculate the probability that agent $j_B$ has reported $\hat{r}$, we should sum over all possible private signals agent $j_B$ has received and calculate the probability agent $j_B$ reported $\hat{r}$ conditioning on he received private signal $\sigma'$, which is determined by agent $j_B$'s strategy $\bm{\pi}(s_{j_A})$. 
	
	By abusing notation a little bit, we can write ${\bm{\pi}(s_{j_A})(\sigma',\bm{\pi}^{-1} Q_{j_A})}$ as a random variable (it is actually a distribution) with $Pr({\bm{\pi}(s_{j_A})(\sigma',\bm{\pi}^{-1} Q)}=\hat{r})={\bm{\pi}(s_{j_A})(\sigma',\bm{\pi}^{-1} Q)}(\hat{r})$. According to above explanation, (\ref{im1}) follows. 
	
	(\ref{im2}) follows from the definition of permuted strategy.
	
	(\ref{im3}) follows from the definition of permuted prior. 
	
	By replacing $\pi_j(\sigma')$ by $\sigma''$, (\ref{im4}) follows.

	We finished the proof $A\approx B$, as previously argued, result (1) follows.  
	
	\vspace{5pt}
	
	\paragraph{Proof for Part (2)}

	We will prove the second part by contradiction:
	
	Fix permutation strategy profile $\bm{\pi}$. First notice that there exists a positive integer $O_d$ such that $\bm{\pi}^{O_d}=I$ where $I$ is the identity and agents play $I$ means they tell the truth.
	
	Given any strategy profile $s$, for the sake of contradiction, we assume that there exists a mechanism $\mathcal{M}$ with unknown prior profile such that $\nu^{\mathcal{M}}_{i_A}(n,\Sigma,\mathbf{Q},\mathbf{s})>\nu^{\mathcal{M}}_{i_A}(n,\Sigma,\mathbf{Q},\bm{\pi}(\mathbf{s}))$ for any prior $Q$. For positive integer $k\in\{0,1,...,O_d\}$, we construct three scenarios:
	
	$$A_k:=(\bm{\pi}^k(\mathbf{Q}),\mathbf{s}),\ A_{k+1}:= (\bm{\pi}^{k+1}(\mathbf{Q}),\mathbf{s}),\ B_{k}:=(\bm{\pi}^k(\mathbf{Q}),\ \bm{\pi}(\mathbf{s}))$$
	
	and show for any $k$, 
	
	(I)$\nu^{\mathcal{M}}_{i_A}(n,\Sigma,A_{k})>\nu^{\mathcal{M}}_{i_A}(n,\Sigma,B_{k})$, 
	
	(II) $\nu^{\mathcal{M}}_{i_A}(n,\Sigma,A_{k+1})=\nu^{\mathcal{M}}_{i_A}(n,\Sigma,B_{k})$. 
	
	Combining (I), (II) and the fact $A_0=A_{O_d}$, we have $$\nu^{\mathcal{M}}_{i_A}(n,\Sigma,A_{0})>\nu^{\mathcal{M}}_{i_A}(n,\Sigma,A_{1})>...\nu^{\mathcal{M}}_{i_A}(n,\Sigma,A_{O_d})=\nu^{\mathcal{M}}_{i_A}(n,\Sigma,A_{0})$$ which is a contradiction. 
	
	Now it is only left to show (I) and (II). Based on our assumption $$\nu^{\mathcal{M}}_{i_A}(n,\Sigma,\mathbf{Q},\mathbf{s})>\nu^{\mathcal{M}}_{i_A}(n,\Sigma,\mathbf{Q},\bm{\pi}(\mathbf{s}))$$ for any prior $\mathbf{Q}$, we have (I). By the same proof we have in part (1), we have $A_{k+1}\approx B_{k}$, which implies (II) according to our above observations. 
	
	\vspace{5pt}
	When the mechanism knows the prior is symmetric, the above proof is still valid if we replace the permutation list $\bm{\pi}$ by symmetric permutation list $(\pi,\pi,...,\pi)$. 
\end{proof}


    
\bibliographystyle{plainnat}
\bibliography{ref}

\appendix

\section{Independent Work Analysis}\label{sec:iw}
 
The analysis in Section~\ref{sec:reprovemd} is not restricted to \citet{dasgupta2013crowdsourced}. Replacing the $R_{i,j}^k$ defined in $M_d$ by the $R_{i,j}^k$ defined in the non-binary extension of $M_d$ in the independent work of \citet{2016arXiv160303151S} will not change the analysis. Thus, \citet{2016arXiv160303151S} is also a special case of $f$-mutual information mechanism---$\mathcal{M}_{MI^{tvd}}$ in the non-binary settings.


\paragraph{The Correlated Agreement (CA) Mechanism \cite{2016arXiv160303151S}} With Assumption~\ref{assume:pc}, the non-binary extension of $M_d$---the CA mechanism---can be reinterpreted as $M_d$ by defining $$R_{i,j}^k:=\mathbbm{1}(\hat{\sigma}_i^k=\hat{\sigma}_j^k)-\mathbbm{1}(\hat{\sigma}_i^{\ell_A}=\hat{\sigma}_j^{\ell_B})$$ where $\ell_A$ is picked from subset $A$ uniformly at random and $\ell_B$ is picked from subset $B$ uniformly at random. 

With this new definition of $R_{i,j}^k$, Claim~\ref{claim:1} is still valid since the proof of Claim~\ref{claim:1} that uses the definition of $R_{i,j}^k$---  \begin{align*}
      &\sum_{\sigma} \left(\Pr[{\Psi}_i=\sigma,{\Psi}_j=\sigma]-\Pr[{\Psi}_i=\sigma]\Pr[{\Psi}_j=\sigma]\right)\\  \tag{Definition of $R_{i,j}^k$ in $M_d$}
      &=\E[R_{i,j}^k]  
 \end{align*}---is still valid for this new definition of $R_{i,j}^k$. Therefore, Theorem~\ref{md} is still valid when we replace $M_d$ by the CA mechanism which means we can also use our information theoretic framework to analyze \citet{2016arXiv160303151S}.

\section{Proof of the Truthfulness of BTS}
\begin{proof}[Proof of Theorem~\ref{bts} (i) \cite{prelec2004bayesian}]

When everyone else tells the truth, for every $i$, agent $i$ will report truthful $\mathbf{p}_i$ to maximize her expected prediction score based on the properties of log scoring rule. Thus, $\hat{\mathbf{p}}_i=\mathbf{p}_i$ for every $i$. 

For the expected information score, we want to calculate the optimal $\sigma$ agent $i$ should report to maximize her expected information score when everyone else tells the truth, given that she receives $\Psi_i=\sigma_i)$. 

\begin{align*}
    &\argmax_{\sigma}\E_{\Psi_j,W|\Psi_i=\sigma_i}\log(\frac{Pr[{\Psi}_i=\sigma |W]}{Pr[{\Psi}_i=\sigma|\Psi_j]})\\ \tag{Conditional independence}
    &=\argmax_{\sigma}\E_{\Psi_j,W|\Psi_i=\sigma_i} \log(\frac{Pr[{\Psi}_i=\sigma|W,\Psi_j]}{Pr[{\Psi}_i=\sigma|\Psi_j]}) \\
     &=\argmax_{\sigma}\E_{\Psi_j,W|\Psi_i=\sigma_i} \log(\frac{Pr[{\Psi}_i=\sigma,W|\Psi_j]}{Pr[{\Psi}_i=\sigma|\Psi_j]Pr[W|\Psi_j]}) \\ 
    &=\argmax_{\sigma}\E_{\Psi_j,W|\Psi_i=\sigma_i} \log(\frac{Pr[W|{\Psi}_i=\sigma,\Psi_j]}{Pr[W|\Psi_j]}) \\ \tag{we can add $\log \Pr[W|\Psi_j]$ which is independent of $\sigma$}
    &=\argmax_{\sigma}\E_{\Psi_j}L(\Pr[\bm{W}|\Psi_i=\sigma_i,\Psi_j],\Pr[\bm{W}|\Psi_i=\sigma,\Psi_j]) \\ \tag{$\argmax_{\mathbf{q}}L(\mathbf{p},\mathbf{q})=\mathbf{p}$}
    &= \sigma_i
\end{align*}

Therefore, in BTS, for every $i$, agent $i$'s best response is $(\sigma_i,\mathbf{p}_i)$ when everyone else tells the truth. BTS is truthful. 

\end{proof}


\section{Proof of Information Monotonicity of $f$-divergence}
{
\renewcommand{\thetheorem}{\ref{lem:im}}
\begin{fact}[Information Monotonicity (\cite{ali1966general,liese2006divergences,amari2010information})]
For any strictly convex function $f$,  $f$-divergence $D_f(\mathbf{p},\mathbf{q})$ 
satisfies information monotonicity so that for any transition matrix $\theta \in \mathbbm{R}^{|\Sigma| \times |\Sigma|}$, $D_f(\mathbf{p},\mathbf{q})\geq D_f(\theta^T \mathbf{p},\theta^T \mathbf{q})$. 

Moreover, the inequality is strict if and only if there exists $\sigma, \sigma',\sigma''$ such that $\frac{\mathbf{p}(\sigma'')}{\mathbf{p}(\sigma')}\neq \frac{\mathbf{q}(\sigma'')}{\mathbf{q}(\sigma')}$ and $\theta_{\sigma',\sigma} \mathbf{p}(\sigma')>0$, $\theta_{\sigma'',\sigma} \mathbf{p}(\sigma'')>0$.
\end{fact}
\addtocounter{theorem}{-1}
}

If the strictness condition does not satisfied, we can see $\theta^T \mathbf{p}$ and $\theta^T \mathbf{q}$ are $\mathbf{p}$ and $\mathbf{q}$'s sufficient statistic which means the transition $\theta$ does not lose any information, thus, the equality holds.

\begin{proof}
The proof follows from algebraic manipulation and one application of convexity.

\begin{align}
D_f(\theta^T \mathbf{p},\theta^T \mathbf{q})=&\sum_{\sigma} (\theta^T \mathbf{p})(\sigma) f\left(\frac{(\theta^T \mathbf{q})(\sigma)}{(\theta^T\mathbf{p})(\sigma)}\right)\\
=& \sum_{\sigma} \theta^T_{\sigma,\cdot} \mathbf{p} f\left(\frac{\theta^T_{\sigma,\cdot} \mathbf{q}}{\theta^T_{\sigma,\cdot}\mathbf{p}}\right)\\
=& \sum_{\sigma} \theta^T_{\sigma,\cdot} \mathbf{p} f\left(\frac{1}{\theta^T_{\sigma,\cdot}\mathbf{p}}\sum_{\sigma'} \theta^T_{\sigma,\sigma'}\mathbf{p}(\sigma') \frac{\mathbf{q}(\sigma')}{\mathbf{p}(\sigma')}\right)\\ \label{eq_im}
\leq & \sum_{\sigma} \theta^T_{\sigma,\cdot} \mathbf{p} \frac{1}{\theta^T_{\sigma,\cdot}\mathbf{p}}\sum_{\sigma'} \theta^T_{\sigma,\sigma'}\mathbf{p}(\sigma') f\left( \frac{\mathbf{q}(\sigma')}{\mathbf{p}(\sigma')}\right)\\
= & \sum_{\sigma} \mathbf{p}(\sigma) f\left(\frac{\mathbf{q}(\sigma)}{\mathbf{p}(\sigma)}\right)= D_f(\mathbf{p},\mathbf{q})
\end{align}

The second equality holds since $(\theta^T\mathbf{p})(\sigma)$ is dot product of the $\sigma^{th}$ row of $\theta^T$ and $\mathbf{p}$.

The third equality holds since $\sum_{\sigma'} \theta^T_{\sigma,\sigma'}\mathbf{p}(\sigma') \frac{\mathbf{q}(\sigma')}{\mathbf{p}(\sigma')}=\theta^T_{\sigma,\cdot}\mathbf{q}$. 

The fourth inequality follows from the convexity of $f(\cdot)$. 

The last equality holds since $\sum_{\sigma}\theta^T_{\sigma,\sigma'}=1$. 

We now examine under what conditions the inequality in Equation~\ref{eq_im} is strict. 
Note that for any strictly convex function $g$, if $\forall u , \lambda_u>0$, $g(\sum_u \lambda_u x_u)=\sum_u \lambda_u g(x_u)$ if and only if there exists $x$ such that $\forall u , x_u=x$. By this property, the inequality is strict if and only if there exists $\sigma, \sigma',\sigma''$ such that $\frac{\mathbf{q}(\sigma')}{\mathbf{p}(\sigma')}\neq \frac{\mathbf{q}(\sigma'')}{\mathbf{p}(\sigma'')}$ and $\theta^T_{\sigma,\sigma'} \mathbf{p}(\sigma')>0$, $\theta^T_{\sigma,\sigma''} \mathbf{p}(\sigma'')>0$.

\end{proof}

\end{document}